\documentclass[journal,twocolumn]{IEEEtran}

\usepackage{amsmath,amssymb,amsthm}

\usepackage{verbatim, color}
\usepackage{cite, hyperref}

\usepackage{graphicx}

\newtheorem{corollary}{Corollary}
\newtheorem{definition}{Definition}

\newtheorem{lemma}{Lemma}
\newtheorem{proposition}{Proposition}
\newtheorem{theorem}{Theorem}

\newtheorem*{remark}{Remark}

\newcommand{\bra}[1]{\langle #1|}
\newcommand{\ket}[1]{|#1\rangle}
\newcommand{\ketbra}[2]{|#1\rangle\langle #2|}

\newcommand{\ip}[2]{\langle #1|#2\rangle}
\newcommand{\op}[2]{|#1\rangle \langle #2|}
\newcommand{\1}{{\openone}}

\newcommand{\mbf}{\mathbf}
\newcommand{\mbb}{\mathbb}
\newcommand{\mc}{\mathcal}

\newcommand{\tr}{\text{Tr}}

\DeclareMathOperator{\diag}{\textrm{diag}}

\begin{document}

\title{Bounds on Instantaneous Nonlocal Quantum~Computation}

\author{Alvin~Gonzales and Eric~Chitambar\\

\thanks{A.G. is with the Department of Computer Science, Southern Illinois University, Carbondale, Illinois 62901, USA.}

\thanks{E.C. is with the Department of Electrical and Computer Engineering, Coordinated Science Laboratory, University of Illinois at Urbana-Champaign, Urbana, IL 61801, USA.}}

\maketitle

\begin{abstract}

Instantaneous nonlocal quantum computation refers to a process in which spacelike separated parties simulate a nonlocal quantum operation on their joint systems through the consumption of pre-shared entanglement.  To prevent a violation of causality, this simulation succeeds up to local errors that can only be corrected after the parties communicate classically with one another.  However, this communication is non-interactive, and it involves just the broadcasting of local measurement outcomes.  We refer to this operational paradigm as local operations and broadcast communication (LOBC) to distinguish it from the standard local operations and (interactive) classical communication (LOCC). 

%LOBC emerges as the appropriate class of operations to consider in time-sensitive quantum information tasks, such as quantum position verification (QPV).  
 
In this paper, we show that an arbitrary two-qubit gate can be implemented by LOBC with $\epsilon$-error using $O(\log(1/\epsilon))$ entangled bits (ebits).  This offers an exponential improvement over the best known two-qubit protocols, whose ebit costs behave as $O(1/\epsilon)$.  We also consider the family of binary controlled gates on dimensions $d_A\otimes d_B$.  We find that any hermitian gate of this form can be implemented by LOBC using a single shared ebit.  In sharp contrast, a lower bound of $\log d_B$ ebits is shown in the case of generic (i.e. non-hermitian) gates from this family, even when $d_A=2$.  This demonstrates an unbounded gap between the entanglement costs of LOCC and LOBC gate implementation.  Whereas previous lower bounds on the entanglement cost for instantaneous nonlocal computation restrict the minimum dimension of the needed entanglement, we bound its entanglement entropy.  To our knowledge this is the first such lower bound of its kind.

%Any nonlocal operation can be performed with arbitrarily high fidelity using LOBC, although the best known protocols consume an amount of entanglement that scales exponentially in the system size.  Whether this exponential overhead is necessary remains an important open problem, and the best lower bounds only require the dimension of the entanglement resource to scale linearly with respect to the dimension of the simulated gate.  A related problem fixes the dimension of the gate and asks how the entanglement cost scales as a function of the simulation error $\epsilon$.

\end{abstract}

\section{Introduction}
Distributed quantum computing on a multipartite system can arise in many common scenarios. For example, individuals at two different countries communicating classically with each other might want to combine their computing power to solve a difficult problem together. This type of quantum computation has been studied extensively under the setting of local operations and classical communication (LOCC).  Under LOCC, pre-shared entanglement can be manipulated and put to use in some quantum information processing task.  In particular, the parties can transmit quantum states back and forth using teleportation \cite{Bennett-1993a}, and thus they can simulate any quantum gate that acts globally across their systems.

In this paper, we focus on the setting of local operations and broadcast communication (LOBC). Contrary to the standard LOCC model, in LOBC the classical communication is non-interactive, meaning the parties can just send each other one message that depends only on their own local measurement data.  Hence, consecutive rounds of teleportation are forbidden in this model.  Research into LOCC has typically made a distinction between protocols in which just a single party sends a message (i.e. one-way protocols) and those in which interactive messages are exchanged between the parties (i.e. two-way protocols).  More generally, the subject of LOCC round complexity studies the question of how much more powerful LOCC operations become as more rounds of classical communication are permitted \cite{Owari-2008a, Chitambar-2011a, Nathanson-2013a, Wakakuwa-2016a, Chitambar-2017a}.  %LOBC can be viewed as a certain hybrid of one-way and two-way LOCC protocols. 

There are two main motivations for considering LOBC operations.  The first, being practical in nature, is that an LOBC protocol is typically more time efficient than a general LOCC process.  More precisely, the duration of an LOBC protocol is no longer than the time it takes a message to be sent between two parties of greatest separation.  This is of vital importance for realistic quantum information processing in which maintaining coherence for long time lengths is a formidable challenge.  The time-constrained nature of LOBC processing has also found cryptographic application in the task of position verification \cite{Chandran-2009a, Kent-2011a, Malaney-2010a, Lau-2011a, Buhrman-2013a}, and we review this connection in Section \ref{Sect:QPV}. 

A second motivation is more fundamental in nature and it involves understanding \textit{interaction} as a resource in distributed quantum information.  The specific problem we study in this paper is the simulation of some nonlocal gate using pre-shared entanglement and LOBC operations.  Historically, this task has been referred to as \textit{instantaneous nonlocal computation}, but such a title can be misleading as the complete computation requires a nonzero implementation time; see Section \ref{Sect:NLQC}.  We consider the question of how much entanglement is needed to simulate a given gate when non-interactive classical communication is allowed.  This LOBC entanglement cost can then be compared to the LOCC entanglement cost of simulating the same gate when interactive classical communication is permitted (see Figs. \ref{fig:LOCC-simulation} and \ref{fig:LOBC-simulation}).  As a result, quantitative trade-offs can be formulated between shared entanglement and interactive classical communication.  Beyond exemplifying this type of resource trade-off, the task of instantaneous nonlocal computation touches on foundational questions in computation theory, as it provides a benchmark for assessing operational capabilities in generalized probability theories \cite{Buhrman-2014a, Broadbent-2016a}.

This paper is structured as follows.  We begin in the next section by describing the task of instantaneous nonlocal computation.  Known results are reviewed and they are compared to analogous results in the general LOCC setting.  In Section \ref{Sect:QPV}, the cryptographic application of position verification is described in both the classical and quantum settings.  Section \ref{Sect:Results} contains our new results which involve deriving improved upper and lower bounds on the entanglement cost of simulating different families of gates using LOBC.  The main proofs and protocols are then presented in Section \ref{protocols}, and finally Section \ref{Sect:Conclusion} provides some concluding remarks.

%For example, using LOCC, Alice can teleport her qubit state to Bob, Bob applies the gate locally, and Bob teleports Alice's qubit state back to Alice. This uses two rounds of communication. An LOBC implementation would take half of the time because Alice and Bob only broadcast information at the end. This is why computations of this type are called instantaneous nonlocal quantum computation \cite{Buhrman-2014a}. This efficiency in time leads to interesting applications. However, this often (not always) comes at the cost of more entanglement consumption \cite{Buhrman-2014a, Beigi-2011a}. 

\section{Instantaneous Nonlocal Quantum Computation}

\label{Sect:NLQC}

\begin{figure}
    \centering
        \includegraphics[width=.48\textwidth]{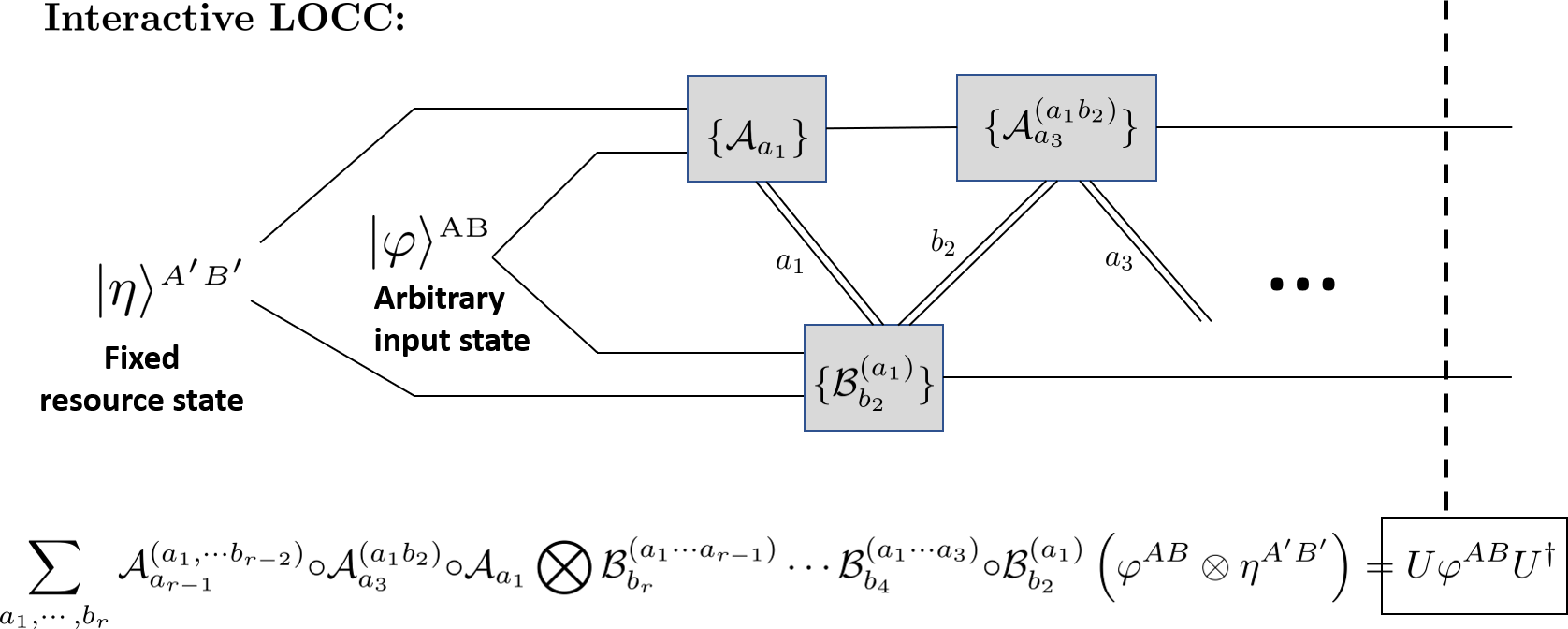}
        \caption{The LOCC simulation of a nonlocal gate $U$ may involve multiple rounds of interactive communication (see, for example, \cite{Wakakuwa-2016a}).  Alice and Bob perform local measurements and communicate their measurement outcomes $a_n$ and $b_{n+1}$.  The choice of local measurement at each round can depend on the outcomes of previous measurements.}
        \label{fig:LOCC-simulation}
    ~ %add desired spacing between images, e. g. ~, \quad, \qquad, \hfill etc. 
    \end{figure}

\begin{figure}
    \centering
        \includegraphics[width=.46\textwidth]{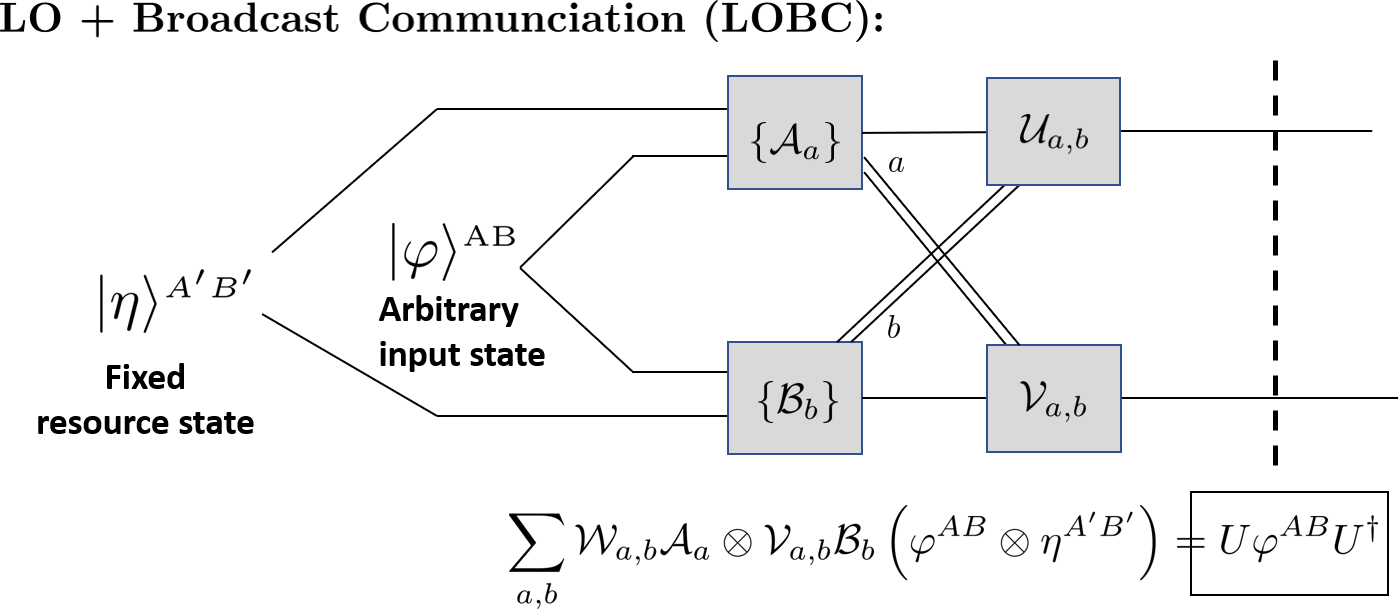}
        \caption{In the LOBC simulation of a nonlocal gate $U$, two-way signaling is allowed but with no interaction.  Protocols of this form are called instantaneous nonlocal computation of the gate $U$.  This paper considers how much more entanglement $\ket{\eta}$ is needed in the LOBC model to make up for the lost interactive classical communication.}
        \label{fig:LOBC-simulation}
    ~ %add desired spacing between images, e. g. ~, \quad, \qquad, \hfill etc. 
    \end{figure}

In instantaneous nonlocal quantum computation, the goal is to apply a global unitary gate over some multipartite system using local measurements alone.  That is, for a given unitary $U$ and arbitrary initial state $\ket{\psi}$, i.e., one whose classical description is unknown to the parties, they wish to invoke the transformation 
\begin{equation}
\ket{\psi}\to U\ket{\psi}
\end{equation}
by performing simultaneous local measurements on their respective subsystems; hence the description ``instantaneous nonlocal computation.''  Of course, the notion of ``instantaneous computation'' should not be taken literally since this process is not physically possible for two reasons.  The first reason is that $U$ may be an entangling gate, and the transformation $\ket{\psi}\to U\ket{\psi}$ could then generate entanglement, something which is not possible using local operations.  One can overcome this objection by allowing the parties to consume entanglement in the process.  Such a transformation then takes the form 
\begin{equation}
\ket{\psi}\otimes\ket{\eta}\to U\ket{\psi},
\end{equation}
where $\ket{\eta}$ is some pre-shared entanglement resource known to all the parties.  However, this process is still not possible in general due to relativistic constraints.  If, for example, $U$ were simply a permutation operators, then the transformation $\ket{\psi}\otimes\ket{\eta}\to U\ket{\psi}$ could allow for instantaneous communication among the spatially separated parties, an impossibility even when using an unbounded amount of entanglement $\ket{\eta}$ \cite{Bennett-2003a}.  Thus the problem must be further modified if it is to be physically feasible.  

One relaxation is to allow for locally correctable errors on the transformed state.  The collective outcomes of the different local measurements can be denoted by variable $m$ so that given particular outcomes $m$, the induced state transformation has the form $\ket{\psi}\to\ket{\phi_m}$.  Instead of aiming to achieve $\ket{\phi_m}=U\ket{\psi}$ for every $m$, the goal is for $\ket{\phi_m}\overset{LU(m)}{=} U\ket{\psi}$, where $\overset{LU(m)}{=}$ means that the two states are related by a local unitary (LU) transformation that can be determined from the measurement data $m$.  In this sense, the task of instantaneous nonlocal quantum computation of the gate $U$ means that
\begin{equation}
\label{Eq:INLQC}
\ket{\psi}\otimes \ket{\eta}\to\ket{\phi_m}\overset{LU(m)}{=} U\ket{\psi}\qquad\forall m,
\end{equation}
using local quantum measurements having outcomes $m$.  This could be further relaxed by considering target states $\epsilon$-close to $U\ket{\psi}$ or by allowing the equality to hold not for all measurement outcomes $m$, but only those belonging to some highly probable set.  Equation \eqref{Eq:INLQC} thus describes a process using local operations and broadcast communication (LOBC).  Each party makes a suitable local measurement and then broadcasts the outcome.  From this globally shared information $m$, the LU error correction can be determined and implemented with no further communication.  The resultant transformation is then $\ket{\psi}\otimes\ket{\eta}\to U\ket{\psi}$, and the desired simulation of gate $U$ is achieved.  The main focus of this paper is on determining the minimal amount of entanglement $\ket{\eta}$ needed to simulate a given unitary $U$ in this way.

That it is even possible to perform Eq. \eqref{Eq:INLQC} for every unitary $U$ is not obvious.  It was first shown by Vaidman \cite{Vaidman-2003a} that instantaneous nonlocal computation can always be attained with arbitrarily high probability provided that the parties share enough entanglement.  Specifically, in Vaidman's scheme the entanglement consumption scales as $O(2^{\log(1/\epsilon)\cdot 2^{4n}})$, with $\epsilon$ being the error and $n$ being the number of qubits comprising the shared state $\ket{\psi}$.  In this protocol, the full entanglement $\ket{\eta}$ must be consumed for every outcome $m$.  An improved protocol was devised by Clark \textit{et al.} in which some of the outcomes $m$ use only part of the initial entanglement, leaving the remainder usable for another task \cite{Clark-2010a}.  However, the average entanglement consumed across all outcomes $m$ in this protocol still scales double exponentially in the system size.  A breakthrough was later made by Beigi and K\"{o}nig who used port-based teleportation \cite{Ishizaka-2008a, Ishizaka-2009a} as a primary subroutine within their protocol \cite{Beigi-2011a}.  They were able to develop a general method for instantaneous nonlocal computation that uses only $O(n\frac{2^{8n}}{\epsilon^2})$ ebits.  

Subsequent work has also been conducted on the instantaneous nonlocal computation of certain families of gates.  For gates belonging to the so-called Clifford hierarchy, specialized protocols have been devised by Chakraborty and Leverrier \cite{Chakraborty-2015a}.  General LOBC protocols were referred to as fast protocols by Yu et al. in Ref. \cite{Yu-2012a}, and they were able to construct specific protocols for the nonlocal implementation of unitaries having certain group structure.  A different resource analysis has been carried out by Speelman who related entanglement consumption to the $T$-gate configuration in a quantum circuit realizing a given unitary $U$ \cite{Speelman-2016a}.  A restricted form of LOBC operations were studied for the task of entanglement distillation under the name of ``measure and exchange'' (MX) operations \cite{Rozpedek-2018a}. 

An important problem in the study of instantaneous nonlocal computation is to prove lower bounds on the entanglement cost for implementing certain gates.  One automatic lower bound comes from the \textit{entangling power} of the gate, which was alluded to at the start of this section.  The entangling power is defined as the maximum increase in entanglement among all input states acted upon by the gate, and entanglement monotonicity under LOCC prohibits the entanglement implementation cost from being less than the entangling power.  Note that since the entangling power is a property of the gate, it cannot be used as a lower bound that differentiates the LOCC and LOBC entanglement costs of implementation.  Unfortunately, beyond the entangling-power bound, relatively little else has been proven.  While the best upper bounds for simulating an arbitrary gate have entanglement costs that scale exponentially in the system size, it is unknown whether this amount of entanglement is necessary.  The best lower bounds on the dimension of the shared entanglement scale linearly in the system dimension of the gate being implemented \cite{Beigi-2011a, Tomamichel-2013a}.  A similar lower bound was proven for a BB84-based gate except in terms of the entanglement measure $E_{\max}$ \cite{Ribeiro-2015a}.  One drawback of these lower bounds is that they are not given in terms of ebit cost, unlike the upper bounds.  This can be problematic for making comparative statements between upper and lower bounds.  For example, if one considers the measure $E_{\max}$, which is no greater than the dimension of the entanglement, then the family of states 
\begin{equation}
\ket{\eta_d}=\sqrt{1-\frac{1}{\sqrt{d}}}\ket{11}+ \sqrt{\frac{1}{\sqrt{d}(d-1)}}\sum_{k=2}^{d}\ket{kk}
\end{equation}
demonstrates $E_{\max}(\op{\eta_d}{\eta_d})\to\infty$ as $d\to\infty$, while $\mathrm{E}(\op{\eta_d}{\eta_d})\to 0$.  Here $\mathrm{E}$ is the entanglement entropy which quantifies the amount of ebits in a bipartite pure state \cite{Bennett-1996a, Popescu-1997a}.  The divergence of $E_{\max}$ in this example can be easily seen from the fact that $E_{\max}(\op{\eta_d}{\eta_d})$ coincides with the log-robustness of entanglement \cite{Datta-2009b}, which has the form $2\log(\sum_{k=1}^d\lambda_k)$ for Schmidt coefficients $\lambda_k$.  Thus, $E_{\max}$ and the entanglement entropy $\mathrm{E}$ can behave quite differently, and in terms of ebit cost, no lower bounds have been previously demonstrated for instantaneous nonlocal computation beyond the entanglement power.  To our knowledge, the same is also true for general LOCC gate simulation.

This is particularly relevant to the question of trade-offs between entanglement and interaction described in the introduction.  One motivation for this work is to understand classical interaction as a resource in distributed quantum information processing.  Its resource character can be quantified in terms of how much entanglement the parties must spend to remove interaction from the general LOCC setting and still complete the given task.  Hence, it seems very natural to make this quantification using the standard resource unit of entanglement, which is an ebit.  In this paper we provide such an ebit lower bound on the entanglement cost of performing generic bipartite controlled-phase gates using LOBC (Theorem \ref{Thm:2bys}).

To make a comparison between protocols with interactive communication and those without, we now briefly review some relevant results on the task of gate simulation using general LOCC.  First note that any $d_A\times d_B$ gate can be implemented using teleportation and interactive communication at a cost of $2\log d_A$ ebits.  However, often this is not the optimal protocol.  For Clifford gates, the entanglement cost is to equal the entangling power \cite{Wakakuwa-2019a}, which can be less than the dimension-bound of teleporation.  For two qubits, any controlled unitary gate can be implemented under LOCC with just one shared ebit and two bits of classical information \cite{Collins-2001a, Eisert-2000a}.  This entanglement cost was later proven to be optimal for resource states having Schmidt rank two \cite{Soeda-2011a}.  A generalization of this result came in Ref. \cite{Stahlke-2011a}, where it was shown that if an entangled resource state can simulate a unitary gate whose Schmidt rank is the same as the resource state, then the latter must be maximally entangled.  Interestingly, these lower bounds no longer hold for resource states having a Schmidt rank that exceeds the Schmidt rank of the simulated gate, and they therefore fail to provide an ebit lower bound on the LOCC entanglement cost of gate simulation.
%More generally, Eisert \textit{et al.} \cite{Eisert-2000a} have shown that an $N$ party controlled unitary gate can be implemented using $N-1$ shared ebits and $2(N-1)$ bits of classical communication. In the general case, $2(N-1)$ ebits is sufficient to implement an arbitrary unitary gate over an $N$ party system \cite{Chefles-2001a}.
In complementary earlier work, Cirac \textit{et al.} have shown that the entanglement needed to simulate a family of weakly entangling gates can be smaller than one, and it approaches zero as the entangling power of these gates likewise approaches zero \cite{Cirac-2001a}.  Our main protocol in Theorem \ref{Thm:2qubits} draws inspiration from the protocol described in Ref. \cite{Cirac-2001a}. 

When studying the entanglement cost of implementing a nonlocal unitary using either LOBC or LOCC, the problems of exact simulation versus $\epsilon$-approximate simulation are different in nature.  In fact, the entanglement cost could be far less in the $\epsilon$-approximate regime, and arguably this is the more relevant setting to consider for realistic applications.  However, the problem of exact simulation is still important from a fundamental perspective as it allows for fundamental separations to be drawn between LOBC and general LOCC.  Furthermore, if one places a bound on the dimension of the entanglement resource, then the set of LOBC operations is compact and the cost of exact simulation serves as a limit for the $\epsilon$-approximate cost as $\epsilon\to 0$.  In this paper we consider both variants of the problem.  Specifically, Theorem \ref{Thm:2qubits} pertains to the approximate simulation of an arbitrary two-qubit gate whereas Proposition \ref{Prop:2ebitsExact}, Theorem \ref{Thm:cHermitian}, and Theorem \ref{Thm:2bys} deal with exact implementations.

\section{Classical and Quantum Position Verification}

\label{Sect:QPV}

A concrete application of instantaneous nonlocal quantum computation by LOBC is quantum position verification (QPV). In position verification, a group of verifiers want to check if a prover $P$, who claims to be in position $pos$, is indeed at that location. A general verification scheme is to send a challenge to $P$ and check if $P$ responds with the correct answer within a specified amount of time. This technique is called distance bounding, and it was introduced in the classical setting by Brands and Chaum \cite{Brands-1993a}.  The intuition behind the scheme is that the adversaries, none of whom are at $pos$, are prohibited by relativistic constraints to correctly respond to the challenge within the allowed time frame.
%In the classical setting, several wireless position verification schemes were introduced \cite{Sastry-2003a, Vora-2006a, Capkun-2005a, Singelee-2005a, Zhang-2006a, Capkun-2006a}. 
However, this intuition fails, and classical position verification has been shown to be insecure against multiple colluding adversaries \cite{Chandran-2009a}. 

One key step in the classical attacks is the cloning of information by the colluding adversaries.  Since general cloning is not allowed in quantum mechanics, scientists attempted to build secure position-verification protocols based on the exchange of quantum information.
%turned to QPV to try and utilize the property that unknown quantum states cannot be cloned. The intuition behind this strategy is that a group of adversaries, none of whom are at position $pos$, can respond correctly because they need interactive communication.
The first QPV protocols were invented in 2002 under the name ''quantum tagging''  \cite{Kent-2011a} with independent schemes proposed in Refs. \cite{Malaney-2010a} and \cite{Chandran-2010a}.  However, these protocols are insecure provided the attackers have enough pre-shared entanglement \cite{Kent-2011a, Lau-2011a}.
%Kent \textit{et al.} \cite{Kent-2011a} described three protocols that were insecure against teleportation attaks and two protocols that were thought to be secure. Independently, Malaney also introduced a protocol that was argued to be secure \cite{Malaney-2010a}.  However, Lau \textit{et al.} demonstrated the insecurity of these protocols provided that the attackers have enough preshared entanglement \cite{Lau-2011a}. 
In general, all these protocols fall to a general attack based on instantaneous nonlocal quantum computation, as presented in detail by Buhrman \textit{et al.} \cite{Buhrman-2014a}.  The attack relies on teleport$^*$ (teleportation without communication) and the use of multiple ``teleportation'' channels for each possible Pauli error. Thus, at the end of the protocol, the adversaries share the correct state in one of the channels.  Through broadcasting their measurement outcomes, they can then identify this channel and fool the verifiers.  However, the amount of entanglement consumed in this strategy is doubly exponential in the size of the system. Beigi \textit{et al.} \cite{Beigi-2011a} later improved on this result by using "port-based teleportation," which uses an amount of entanglement only exponential in the system size.  It remains an important open problem whether or not QPV attacks exist that are sub-exponential in their entanglement consumption, and the best lower bounds only require the dimension of the entanglement to scale linearly with respect to the dimension of the simulated gate.  

We should emphasize, however, that an LOBC attack is not the most general attack that can be performed on a QPV scheme.  Indeed, LOBC assumes that the adversaries only communicate with one another classically.  A conceivably more powerful attack allows the adversaries to exchange quantum information during the protocol as well.  In other words, the operational class that encompasses a broader class of QPV attacks consists in local operations and broadcast quantum communication (LOBQC).  We do not consider such a model in this paper.

%A related problem fixes the dimension of the gate and asks how the entanglement cost scales as a function of the simulation error $\epsilon$.  Our first result in the next section focuses on two-qubit attacks.  For these dimensions, we provide a attack protocol whose failure can be made exponentially small while consuming a linear number of ebits.

%While QPV is generally insecure, the amount of entanglement required is exponential in the size of the quantum system being measured. In a paper by Buhrman et al. \cite{Buhrman-2013a}, they introduced the garden hose complexity and relate it QPV. They showed that for certain QPV schemes, the size of the resource state is exponential in the size of classical strings generated by the verifiers.  The hope is to make the task of breaking the system impractical.

%In this paper, we introduce a protocol that is logarithmic in its entanglement consumption as a function of the error for two qubit gates. This is an improvement on previous results by Buhrman et al.  \cite{Buhrman-2014a} and Beigi et al. \cite{Beigi-2011a}.

\section{Results}

\label{Sect:Results}

%The LOBC setting is used in the remainder of this paper. In this paper, we introduce a new arbitrary two qubit unitary gate protocol. The protocol's entanglement consumption is a logarithmic function of the error. We also introduce a protocol for implementing a controlled Hermitian unitary over a $d_A\otimes d_B$ system. It is somewhat surprising that this protocol only consumes one ebit despite the sizes of the two systems. Finally, we prove that their exists controlled unitaries on $2\otimes s$ system that require at least $\log(d)$ ebits to implement. Under LOCC we can always implement this with two ebits by using teleportation. Thus, there is an unbounded gap between LOBC and LOCC; and non-Hermitian and Hermitian controlled unitary gates.

\subsection{Two-qubit gates}

\subsubsection{Exact Implementations}

We begin by describing a simple protocol that provides an exact implementation of certain two-qubit unitaries.
\begin{definition}
Let {\upshape $\textbf{L}$} be the family of two-qubit unitaries such that $U\in\;${\upshape $\textbf{L}$} if there exists unitaries $R\otimes\mbb{I}$ and $T_j\otimes V_j$ such that
\begin{equation}
\label{Eq:Unitary-form}
U (R \sigma_j R^\dagger\otimes \mbb{I})U^\dagger=T_j\otimes V_j \qquad\text{for $j=1,2,3$},
\end{equation}
where 
\begin{align}
\sigma_0&=\begin{pmatrix}1&0\\0&1\end{pmatrix}, &\sigma_1&=\begin{pmatrix}0&1\\1&0\end{pmatrix}, \notag\\
\sigma_2&=\begin{pmatrix}0&-i\\i&0\end{pmatrix} &\sigma_3&=\begin{pmatrix}1&0\\0&-1\end{pmatrix}
\end{align}
are the standard Pauli matrices.
\end{definition}

\begin{proposition}\label{Prop:2ebitsExact}
Any $U\in\;${\upshape $\textbf{L}$} can be perfectly simulated by LOBC using two ebits and four classical bits of (non-interactive) communication.
\end{proposition}

\begin{proof}
The protocol we describe for performing $U\in\textbf{L}$ is similar in spirit to the protocol of Vaidman and Buhrman \textit{et al.} \cite{Vaidman-2003a, Buhrman-2013a}, and we call it U2E (the ``E'' in the name stands for ``exact'').  A subroutine in this protocol is teleportation$^*$, which is the standard teleportation protocol except with no classical communication and no Pauli correction on the receiving end \cite{Buhrman-2013a}.  Thus, at the end of teleportation$^*$, the receiver has the teleported state up to a local Pauli error.

\noindent\textbf{Protocol U2E: Two ebit protocol for $U\in\textbf{L}$}

$\bullet$ Input an arbitrary two-qubit state $\ket{\psi}^{AB}$. 
\begin{enumerate}
\item Suppose that $U$ satisfies Eq. \eqref{Eq:Unitary-form}.  Using ebit $\ket{\Phi^+}^{A_1B_1}=\sqrt{1/2}(\ket{00}+\ket{11})^{A_1B_1}$, Alice teleports$^*$ $A_1$ to Bob by measuring in the rotated Bell basis $\{(R\sigma_j R^\dagger\otimes\mbb{I})\ket{\Phi_1^+}^{AA_1}\}_j$.  This leaves Alice (A) and Bob (B) sharing the state
\begin{align}
(R \sigma_j R^\dagger\otimes\mbb{I})\ket{\psi}^{B_1B},
\end{align}
where $\sigma_j$ is a Pauli error known to Alice.
\item Bob applies the unitary $U$ on systems $B_1B$, and by Eq. \eqref{Eq:Unitary-form} we have
\begin{align}
(T_j\otimes V_j)U\ket{\psi}^{B_1B}.
\end{align}
\item Using ebit $\ket{\Phi^+}^{A_2B_2}$, Bob teleports$^*$ $B_1$ back to Alice, they broadcast their results, and then they perform the necessary local error corrections, i.e. Bob's teleportation Pauli error and as well as $T_j^\dagger\otimes V_j^\dagger$.  In total, Alice and Bob are left in the shared state $U\ket{\psi}$, as desired. 
\end{enumerate}
\end{proof}
In some cases, Protocol U2E is optimal.  For example, consider the swap operator $\mbb{F}\in\textbf{L}$, whose action is given by $\mbb{F}(\ket{\alpha}^A\ket{\beta}^B)=\ket{\beta}^A\ket{\alpha}^B$ for an arbitrary product state $\ket{\alpha}\ket{\beta}$.  Since swap has an entangling power of two ebits (when acting on subsystems $AB$ of the state $\ket{\Phi^+}^{AA'}\otimes \ket{\Phi^+}^{BB'}$), protocol U2E is optimal for the nonlocal simulation of swap.  In fact, it is straightforward to generalize protocol U2E to optimally perform the $d$-dimensional swap operator using teleportation$^*$ with a $d$-dimensional Bell basis.  On the other hand Protocol U2E is sub-optimal for other gates.  For example, CNOT is an element of $\textbf{L}$, and Theorem \ref{Thm:cHermitian} below shows that CNOT can be implemented by LOBC using just one ebit.

%for $2^n\otimes 2^n$ systems, protocol U2E can be applied by Alice and Bob in parallel on each of the $n$ constituent two-qubit systems.  The cost is $2n$ ebits, which again is optimal for simulating the swap gate on this system using LOCC. 

Finally, let us briefly comment on the structure of {\upshape $\textbf{L}$}.  First, observe that {\upshape $\textbf{L}$} is closed under local unitary transformations.  That is, if $U\in\;${\upshape $\textbf{L}$}, then so is $V=(R_1\otimes S_1)U(R_2\otimes S_2)$.  In general, we say that unitaries $U$ and $V$ are locally equivalent if they can be related by local unitaries in this way.  Second, consider the two-qubit Pauli group, $\mc{P}_2=\{\sigma_{j}\otimes\sigma_k\}_{j,k=0}^3\times\{\pm 1,\pm i\}$, as well as its normalizer, $\mc{C}_2=\{U: U g U^\dagger\in\mc{P}_2\;\forall g\in\mc{P}_2\}$.  The latter is typically referred to as the Clifford group, and as easily seen from the definitions, any operator locally equivalent to a Clifford operator also belongs to {\upshape $\textbf{L}$}.  However, somewhat surprisingly, the converse is also true.  
\begin{lemma}
\label{Lem:1}
$U\in\;${\upshape $\textbf{L}$} if and only if there exists local unitaries $R_i\otimes S_i$ such that $(R_1\otimes S_1)U(R_2\otimes S_2)\in\mc{C}_2$.
\end{lemma}
\noindent The proof is provided in Section \ref{Sect-Lem:1}, and we suspect this lemma may also find application in other quantum computation tasks.  One immediate consequence of Lemma \ref{Lem:1} is that Protocol U2E is no stronger in terms of entanglement consumption than the protocol recently given in Ref. \cite{Wakakuwa-2019a}.  In that paper, the authors provide an LOBC protocol for the implementation of any Clifford gate (in arbitrary dimension).  Their protocol differs in that it involves Alice and Bob sharing the Choi state of $U$ as their resource entanglement.  Since for two qubits the entanglement of the Choi state can be less than two ebits, their protocol in general will have a smaller entanglement consumption.  However, the resource state used in the protocol of Ref. \cite{Wakakuwa-2019a} is specific to the gate being simulated, whereas protocol U2E uses a gate-independent resource state.  One could modify the protocol of Ref. \cite{Wakakuwa-2019a} by first equipping Alice and Bob with some fixed two-ebit resource state, and then have them convert this into the Choi state of a given unitary by LOCC.  Doing this would render a protocol very similar to U2E.

\medskip

\subsubsection{Approximate Implementations}

We now turn to the problem of instantaneous nonlocal computation of an arbitrary two-qubit unitary.  We present a new protocol referred to as U2, and its detailed description is given in Section \ref{protocols}. Except for certain angles, protocol U2 is probabilistic.  It involves diagonalizing a two-qubit unitary in the so-called ``magic basis'' (see Eq. \eqref{Eq:magicbasis}) and then expressing this diagonalization as a sequence of simple single and two-qubit gates.  The protocol then involves implementing these gates under the LOBC constraint following the ``angle-doubling'' error correction idea of Ref. \cite{Cirac-2001a}.  One of the key features of our protocol is that it does not use Vaidman's ``tree of teleportation channels'' \cite{Vaidman-2003a, Clark-2010a, Buhrman-2014a}, and we therefore avoid an exponential growth in entanglement cost.  Its performance is reported
 in the following theorem.

%Our two-qubit protocol called U2 is quite lengthy so we only present the theorem here and put the detailed proof and protocol in \ref{protocols}. U2 is a  iterative protocol. First we break down an arbitrary two-qubit unitary into simpler gates and implement each one sequentially. The protocol also relies on the fact that one of the parties, which we call the observer, always knows if the gate has been implemented correctly (i.e. the transformed state is $LU$ equivalent to the required state). When a gate has been implemented correctly, the observer enters a halting scheme that only allows the state to undergo local unitary changes. 

\begin{theorem}\label{Thm:2qubits}
Any two-qubit unitary can be performed under LOBC with probability $(1- 2^{-N})^3$ using a consumption of $8N+1$ ebits.
\end{theorem}
\noindent We can compare the efficiency of protocol U2 to the port-based teleportation scheme of Beigi and K\"{o}nig \cite{Beigi-2011a}.  For a two-qubit gate $U$ and any $\epsilon>0$, their protocol generates a quantum channel $\mc{E}$ which consumes $1+\frac{3\cdot 2^{12}}{\epsilon}$ ebits while achieving an approximation of $U$ quantified by $||\mc{E}-\mc{U}||_\diamond\leq \epsilon$, where $\mc{U}(\rho)=U(\rho)U^\dagger$ and $||\cdot||_\diamond$ is the so-called diamond norm \cite{Kitaev-1997a}.  In the protocol U2, Alice and Bob know when they have perfectly implemented the gate and when they have failed.  In the latter case they can simply replace their state with ``white noise,'' and thus U2 implements the quantum channel $\mc{E}_{U2}(\rho)=p\mc{U}(\rho)+(1-p)(\mbb{I}\otimes\mbb{I})/4$ at the cost of $8N+1$ ebits and with $p=(1-2^{-N})^3$.  Setting $\epsilon=2(1-p)$, a straightforward calculation shows
\begin{equation}
||\mc{E}_{U2}-\mc{U}||_\diamond\leq \epsilon
\end{equation}
while consuming 
\begin{equation}
1-8\log[1-(1-\tfrac{\epsilon}{2})^{1/3}]\leq 8\log\left(\frac{1}{\epsilon}\right)+22
\end{equation}
ebits. Hence in terms of approximation error $\epsilon$, protocol U2 offers an exponential saving in the entanglement cost compared to port-based teleportation protocols.  A similar savings holds relative to Vaidman-like schemes \cite{Vaidman-2003a, Buhrman-2013a}.

%\textcolor{red}{For general $d\otimes d$ systems with non-integer $\log d$, the $d$-dimensional swap gate can also be implemented optimally by LOBC. Teleportation and teleportation$^*$ can be performed by using a $d$-dimensional ebit. Using the steps in protocol U2E, the cost would be two $d$-dimensional ebits. Since $d$-dimensional swap has an entangling power of two $d$-dimensional ebits (when acting on subsystems $AB$ of the state $\ket{\Phi_{d}^+}^{AA'}\otimes \ket{\Phi_{d}^+}^{BB'}$), protocol U2E is optimal for the nonlocal simulation of $d$-dimensional swap. Dr. Chitambar can you please check this.}

\subsection{Exact implementation of hermitian binary-controlled gates}

We now turn to a class of unitaries in general $d_A \otimes d_B$ systems.  These are controlled gates of the form
\begin{equation}
\label{Eq:Gate-controlled}
U_c=(\mbb{I}-P)\otimes\mbb{I}+P\otimes V,
\end{equation}
where $P$ is an arbitrary projector on system $A$ and $V=V^\dagger$ is a hermitian unitary operator.  This can be interpreted as a binary switch that applies $V$ on system $B$ when system $A$ lies in the support of $P$.  The LOBC implementation of operators having this form was studied in Ref. \cite{Yu-2012a}.  However, in their protocol the amount of consumed entanglement is not explicitly stated.  Here we show that only a single ebit is needed, regardless of the dimensions.
\begin{theorem}\label{Thm:cHermitian}
Any gate having the form of Eq. \eqref{Eq:Gate-controlled} can be implemented by LOBC using one ebit.
\end{theorem}
\begin{proof}
Let $\ket{\eta}^{A'B'}=\dfrac{1}{\sqrt{2}}(\ket{00}+\ket{11})^{A'B'}$ be a shared ebit.  Alice and Bob perform a generalized measurement with respective Kraus operators $\{A_0,A_1\}$ and $\{B_0,B_1\}$, where 
\begin{align}
\notag A_0&=\left[\left(\mbb{I}-P\right)\otimes\bra{0}+P\otimes\bra{1}\right]^{AA'}\\
\notag A_1&=\left[P\otimes\bra{0}+\left(\mbb{I}-P\right)\otimes\bra{1}\right]^{AA'}\\
\notag B_0&=\dfrac{1}{\sqrt{2}}\left[\mbb{I}\otimes\bra{0}+V\otimes\bra{1}\right]^{BB'}\\
B_1&=\dfrac{1}{\sqrt{2}}\left[\mbb{I}\otimes\bra{0}-V\otimes\bra{1}\right]^{BB'}.
\end{align}
Performing these measurements on the initial state $\ket{\psi}^{AB}\ket{\eta}^{A'B'}$ has outcomes
\begin{align}
\notag A_0B_0&:\quad U_c\ket{\psi}\\
\notag A_0B_1&: \quad [(\mbb{I}-P)\otimes\mbb{I}-P\otimes V]\ket{\psi}\\
\notag A_1B_0&: \quad [P\otimes\mbb{I}+(\mbb{I}-P)\otimes V]\ket{\psi}\\
\notag A_1B_1&: \quad [P\otimes\mbb{I}-(\mbb{I}-P)\otimes V]\ket{\psi}.
\end{align}
Define the unitary operator $Z=(\mbb{I}-P)-P$ on Alice's system.  Then for outcome $A_0B_0$ Alice and Bob do nothing, for outcome $A_0B_1$ they perform $Z\otimes\mbb{I}$, for outcome $A_1B_0$ they perform $\mbb{I}\otimes V$, and for outcome $A_1B_1$ they perform $Z\otimes V$.  This attains $U_c\ket{\psi}$ with probability one. 
\end{proof}

\subsection{An ebit lower bound on the exact implementation of generic binary-controlled gates}
We now consider systems of size $2\otimes s$ and show that, in stark contrast to Theorem \ref{Thm:cHermitian}, there are non-Hermitian controlled unitaries whose ebit consumption for implementation depends on the size of $s$. 
\begin{theorem}\label{Thm:2bys}
Let
\begin{equation}
U_\tau=\sum_{j=0}^{s-1} e^{i\tau_j}\op{j}{j}
\end{equation}
have phase angles $\tau_j\in[0,2\pi)$ such that $\tau_k\not=\tau_l$ for all $k\not=l\in\{0,\cdots,s-1\}$.  An LOBC implementation of the controlled unitary
\begin{equation}
\label{Eq:controlled-unitary}
U_c=\op{0}{0}\otimes \mbb{I}_s+\op{1}{1}\otimes U_\tau
\end{equation}
on a $2\otimes s$ system requires at least $\log s$ ebits of shared entanglement resource.
\end{theorem}
\noindent Note that every controlled gate on $2\otimes s$ controlled from the 2-dimensional side is LU equivalent to $U_c$ in Eq. \eqref{Eq:controlled-unitary}, and generically, the phase angles in $U_\tau$ will be distinct.  The proof of Theorem \ref{Thm:2bys} is presented in Section \ref{protocols}.  It should also be noted that Theorem \ref{Thm:2bys} assumes a pure-state resource, and so the amount of ebits refers to the entanglement entropy of the pure state.  If one considers a mixed-state resource, then the entanglement bound in Theorem \ref{Thm:2bys} refers to the entanglement of formation, which is the average pure-state entanglement entropy minimized over all ensembles realizing the resource state.

What is remarkable about this result is that it not only quantifies a lower bound on nonlocal instantaneous computation in terms of ebits, but it also demonstrates an unbounded gap between LOCC and LOBC. Under interactive LOCC, this gate can easily be performed using two ebits: Alice teleports her system to Bob, he performs $U_c$ on both systems, and then he teleports Alice's qubit back to her.  Hence, Theorem \ref{Thm:2bys} accomplishes one of the main goals of the paper; a rigorous trade-off has been identified between interactive communication and entanglement consumption.

% we can always use teleportation to apply the unitary with 2 ebits. Thus, interactive communication allows an unbounded amount of savings in entanglement consumption. Furthermore, this same gap in ebit consumption exists between Hermitian and non-Hermitian controlled unitary gates. It is unknown (at least to us) why there is an unbounded difference in ebit consumption.

\section{Detailed Proofs and Protocols}\label{protocols}

\subsection{The Two-qubit ``Magic Basis'' and the Proof of Lemma \ref{Lem:1}}

\label{Sect-Lem:1}

The magic basis in two qubits \cite{Hill-1997a, Kraus-2001a} is the orthonormal family of states
\begin{align} \label{Eq:magicbasis}
\notag \ket{\Phi_0}&=\frac{1}{\sqrt{2}}(\ket{00}+\ket{11}),\notag\\
\ket{\Phi_1}&=\frac{-i}{\sqrt{2}}(\ket{00}-\ket{11})=-i\sigma_z\otimes\mbb{I}\ket{\Phi_0}\notag\\
\ket{\Phi_2}&=\frac{-i}{\sqrt{2}}(\ket{01}+\ket{10})=-i\sigma_x\otimes\mbb{I}\ket{\Phi_0},\notag\\
\ket{\Phi_3}&=\frac{1}{\sqrt{2}}(\ket{01}-\ket{10})=-i\sigma_y\otimes\mbb{I}\ket{\Phi_0}.
\end{align}
A number of convenient properties emerge when working in the magic basis, and we review them here since many of our proofs make use of them.

\begin{proposition}
\label{Prop:magic-1}
A unitary $\Omega$ is diagonal in the magic basis, $\Omega=\sum_{k=0}^3 e^{i\phi_k}\op{\Phi_k}{\Phi_k}$, if and only if it can be written as
\begin{equation}
\label{Eq:Pauli-exponential}
\Omega=e^{i(\alpha\sigma_x\otimes\sigma_x+\beta\sigma_y\otimes\sigma_y+\gamma\sigma_z\otimes\sigma_z)},
\end{equation}
where 
\begin{align}
	\phi_0&=\alpha-\beta+\gamma\label{Eq:magic-phase1}\\
	\phi_1&=-\alpha+\beta+\gamma\label{Eq:magic-phase2}\\
	\phi_2&=\alpha+\beta-\gamma\label{Eq:magic-phase3}\\
	\phi_3&=-\alpha-\beta-\gamma.\label{Eq:magic-phase4}
\end{align}
\end{proposition}
\begin{proof}
Note that the $\sigma_k\otimes\sigma_k$ form a pairwise commuting set for $k=x,y,z$.  Thus, we can write
\begin{align}
\label{Eq:Pauli-exponential2}
e^{i(\alpha\sigma_x\otimes\sigma_x+\beta\sigma_y\otimes\sigma_y+\gamma\sigma_z\otimes\sigma_z)}&=e^{i\alpha\sigma_x\otimes\sigma_x}e^{i\beta\sigma_y\otimes\sigma_y}e^{i\gamma\sigma_z\otimes\sigma_z}.
\end{align}
Using the identity $e^{i\theta\sigma_k\otimes\sigma_k}=\cos\theta\mbb{I}+i\sin\theta\sigma_k\otimes\sigma_k$ and the fact that each magic state is an eigenstate of $\sigma_k\otimes\sigma_k$, it follows that 
\[0=\sum_{k=0}^3e^{i\phi_k}\op{\Phi_k}{\Phi_k}-e^{i(\alpha\sigma_x\otimes\sigma_x+\beta\sigma_y\otimes\sigma_y+\gamma\sigma_z\otimes\sigma_z)},\]
iff the $\phi_k$ and $\alpha,\beta,\gamma$ are related according to the above relations.
\end{proof}

\begin{proposition}
\label{Prop:magic-2}
A unitary having the form of Eq. \eqref{Eq:Pauli-exponential} belongs to the Clifford group if and only if $\alpha,\beta,\gamma$ are all multiples of $\pi/4$.
\end{proposition}
\begin{proof}
Let $\Omega$ have the form of Eq. \eqref{Eq:Pauli-exponential}. Suppose that $\alpha,\beta,\gamma$ are all multiples of $\pi/4$.  Then for any $i,j,k$ we will have
\begin{align}
&e^{i\frac{n\pi}{4}\sigma_k\otimes\sigma_k}(\sigma_i\otimes\sigma_j)e^{-i\frac{n\pi}{4}\sigma_k\otimes\sigma_k}\notag\\
&=(\cos(\tfrac{n\pi}{4})\mbb{I}\otimes\mbb{I}+i\sin(\tfrac{n\pi}{4})\sigma_k\otimes\sigma_k)\notag\\
&\qquad\times(\sigma_i\otimes\sigma_j)(\cos(\tfrac{n\pi}{4})\mbb{I}\otimes\mbb{I}-i\sin(\tfrac{n\pi}{4})\sigma_k\otimes\sigma_k)\notag\\
&=\begin{cases}\sigma_i\otimes\sigma_j\quad\text{if $\sigma_i\otimes\sigma_j$ commutes with $\sigma_k\otimes\sigma_k$}\\
\pm (\sigma_i\otimes\sigma_j)\quad\text{if $\sigma_i\otimes\sigma_j$ anti-commutes with $\sigma_k\otimes\sigma_k$}
\end{cases}.
\end{align} 
Hence from Eq. \eqref{Eq:Pauli-exponential2}, we see that
\begin{equation}
\label{Eq:Paui-anti-commute}
e^{i\theta\sigma_k\otimes\sigma_k}(\sigma_i\otimes\sigma_j)e^{-i\theta\sigma_k\otimes\sigma_k}\in\mc{P}_2.
\end{equation}
Conversely, if $\Omega$ is not in the Clifford group, then there must be some $i,j,k$ for which Eq. \eqref{Eq:Paui-anti-commute} does not hold.  This means that $\sigma_i\otimes\sigma_j$ anti-commutes with $\sigma_k\otimes\sigma_k$.  Thus,
\begin{align}
e^{i\theta\sigma_k\otimes\sigma_k}(\sigma_i\otimes\sigma_j)e^{-i\theta\sigma_k\otimes\sigma_k}&=(\sigma_i\otimes\sigma_j)e^{2\theta i\sigma_k\otimes\sigma_k}\notag\\
&=(\sigma_i\otimes\sigma_j)(\cos(2\theta)\mbb{I}\otimes\mbb{I}\notag\\
&\quad+i\sin(2\theta)\sigma_k\otimes\sigma_k),
\end{align}
which clearly belongs to $\mc{P}_2$ whenever $\theta$ is an integer  multiple of $\pi/4$.  As this would be a contradiction, we conclude that $\theta$ cannot be an integer multiple of $\pi/4$.
\end{proof}

\begin{proposition}[\cite{Kraus-2001a}]
\label{Prop:magic-3}
Every two-qubit unitary $U$ is locally equivalent to a matrix diagonal in the magic basis.  That is $U$ can be decomposed as
\begin{equation}
\label{Eq:Magic-state-decompose}
U=(R_1\otimes S_1)\Omega(R_2\otimes S_2),
\end{equation}
where $\Omega$ is diagonal in the magic basis and the $R_i\otimes S_i$ are local unitaries.
\end{proposition}
\noindent A detailed proof is given in Ref. \cite{Kraus-2001a}.  We next make the connection between the magic basis and a gate's ability to generate entanglement.  Here we say that $U$ is non-entangling if $U\ket{\alpha}\ket{\beta}$ is a product state for every $\ket{\alpha}\ket{\beta}$.

\begin{proposition}[\cite{Hill-1997a}]
\label{Prop:magic-4}
A two-qubit unitary is non-entangling iff, up to an overall phase, it is real in the magic basis.  
\end{proposition}

\begin{proof}
From Proposition \ref{Prop:magic-3} we write
\begin{equation}
\label{Eq:Magic-state-decompose-2}
U=(R_1\otimes S_1)\Omega(R_2\otimes S_2)=(R_1\otimes\mbb{I})(\mbb{I}\otimes S_1)\Omega(R_2\otimes\mbb{I})(\mbb{I}\otimes S_2).
\end{equation}
Our argument will involve first showing that every product unitary is real in the magic basis.  Since all product unitaries are non-entangling, Eq. \eqref{Eq:Magic-state-decompose-2} implies that $U$ is non-entangling iff $\Omega$ is non-entangling.  With it having been established that every product unitary is real, the proposition will then follow by showing that $U$ is non-entangling iff $\Omega$ is real in the magic basis.

Let us first consider any operator of the form $\mbb{I}\otimes V$ (or alternatively $V\otimes\mbb{I}$), where $V$ is an arbitrary unitary.  Up to an overall phase, we can always express $V=a\mbb{I}+i\vec{b}\cdot\vec{\sigma}$ with $a\geq 0$ and $\vec{b}$ a vector with real components.  Then
\begin{align}
&\bra{\Phi_i}\mbb{I}\otimes V\ket{\Phi_j}\notag\\
&=a\ip{\Phi_i}{\Phi_j}+\sum_{l=1}^3ib_l\bra{\Phi_i}\mbb{I}\otimes\sigma_l\ket{\Phi_j}\notag\\
&=a\delta_{ij}+\begin{cases}\sum_{l=1}^3ib_l\bra{\Phi_0}\sigma_i\sigma_j\otimes\sigma_l\ket{\Phi_0}\qquad\text{if $i,j>0$}\\
\sum_{l=1}^3b_l\bra{\Phi_0}\sigma_j\otimes\sigma_l\ket{\Phi_0}\qquad\text{if $i=0, j>0$}\\
\sum_{l=1}^3-b_l\bra{\Phi_0}\sigma_i\otimes\sigma_l\ket{\Phi_0}\qquad\text{if $j=0,i>0$}\\
0\qquad\text{if $i=j=0$}.
\end{cases}.
\end{align}
Since $\sigma_i\sigma_j=i\epsilon_{ijk}\sigma_k$, we see that $\mbb{I}\otimes V$ is real when expressed in the magic basis.  Let us write $\Omega=\sum_{k=0}^3e^{i\phi_k}\op{\Phi_k}{\Phi_k}$.  It is straightforward to show that $\ket{\psi}=\sum_{k=0}^3 c_k\ket{\Phi_k}$ is a product state iff $\sum_{k=0}^3c_k^2=0$ \cite{Hill-1997a}.  Since $\Omega\ket{\psi}=\sum_{k=0}^3 c_ke^{i\phi_k}\ket{\Phi_k}$, we have that $\Omega$ is non-entangling iff 
\[0=\sum_{k=0}^3c^2_ke^{i2\phi_k}=e^{i2\phi_0}\sum_{k=0}^3c^2_ke^{i2(\phi_k-\phi_0)}\] whenever $\sum_{k=0}^3c_k^2=0$.  This requires that $\phi_k-\phi_0=\pm \pi$ for all $k$.  In other words, up to an overall phase, $\Omega$ is real.

\end{proof}

We now turn to the proof of Lemma \ref{Lem:1}.  It will make use of one more technical fact.
\begin{proposition}
\label{Prop:Paul-real}
Let $\sigma_k$ be any Pauli operator and $V$ an arbitrary one-qubit unitary.  Then there exists some complex phase $e^{-i\varphi}$ such that $e^{-i\varphi}\tr[V\sigma_i V^\dagger \sigma_k]$ is real for all $i=1,2,3$.  
\end{proposition}
\begin{proof}
We write $\tr[V\sigma_i V^\dagger \sigma_k]=\tr[\sigma_i V^\dagger \sigma_k V]$.  Under unitary conjugation $\sigma_k$ transforms to some other unitary $e^{i\varphi}\hat{n}\cdot\vec{\sigma}$, where $\hat{n}$ is a unit vector with real components and $e^{i\varphi}$ is an overall phase.  Hence $e^{-i\varphi}\tr[V\sigma_i V^\dagger \sigma_k]$ is a real number.

\end{proof}

\newtheorem*{Lem:1proof}{Lemma \ref{Lem:1}}
\begin{Lem:1proof}
$U\in\;${\upshape $\textbf{L}$} if and only if there exists local unitaries $R_n\otimes S_n$ such that $(R_1\otimes S_1)U(R_2\otimes S_2)\in\mc{C}_2$.
\end{Lem:1proof}

\begin{proof}
From the definitions it is clear that if $(R_1\otimes S_1)U(R_2\otimes S_2)\in\mc{C}_2$ then $U\in\;${\upshape $\textbf{L}$}.  To prove the converse, observe that if $U\in\textbf{L}$ then there exists some unitary $V$ such that $\Omega(V \sigma_i V^\dagger\otimes\mbb{I})\Omega^\dagger$ is a product unitary, where, by Proposition \ref{Prop:magic-3},
\begin{equation}
\Omega=\sum_{k=0}^3e^{i\phi_k}\op{\Phi_k}{\Phi_k}
\end{equation} 
is obtained from $U$ by local unitaries.  Hence, it suffices to show that $\Omega(V \sigma_i V^\dagger\otimes\mbb{I})\Omega^\dagger$ being a product unitary for all $i$ implies $\Omega\in\mc{C}_2$.  From this it will follow that  $(R_1\otimes S_1)U(R_2\otimes S_2)\in\mc{C}_2$ for some local unitaries $R_n\otimes S_n$.

If $\Omega(V \sigma_i V^\dagger\otimes\mbb{I})\Omega^\dagger$ is a product unitary it is non-entangling and therefore, by Proposition \ref{Prop:magic-4}, there exists some phase $e^{i\varphi}$ such that
\begin{align}
e^{i\varphi}&\bra{\Phi_j}\Omega(V \sigma_i V^\dagger\otimes\mbb{I})\Omega^\dagger\ket{\Phi_k}\notag\\
&=e^{i\varphi}e^{i(\phi_{j}-\phi_k)}\bra{\Phi_j}V\sigma_i V^\dagger\otimes\mbb{I}\ket{\Phi_k}
\end{align}
is real for each $j$ and $k$.  Under what conditions is this true?  Note that when $j=k$ the component vanishes, and so it suffices to just consider the case of $j\not=k$.  First, suppose that $j=0$ and $k>0$.  Then
\begin{align}
&e^{i\varphi}e^{i(\phi_{j}-\phi_k)}\bra{\Phi_j}V\sigma_i V^\dagger\otimes\mbb{I}\ket{\Phi_k}\notag\\
&=-ie^{i\varphi}e^{i(\phi_{0}-\phi_k)}\bra{\Phi_0}V\sigma_i V^\dagger\sigma_k\otimes\mbb{I}\ket{\Phi_0}\notag\\
&=-ie^{i\varphi}e^{i(\phi_{0}-\phi_k)}\tr[V\sigma_i V^\dagger\sigma_k]\quad\forall i=1,2,3.
\end{align}
If these terms are real for all $i$, then by Proposition \ref{Prop:Paul-real}, there must exist some phase $e^{i\varphi_k}$ such that $ie^{i\varphi}e^{i(\phi_{0}-\phi_k)}e^{i\varphi_k}$ is real.  Hence,
\begin{equation}
\label{Eq:phase-cond-1}
\varphi+\phi_0-\phi_k+\varphi_k=n_{0k}\pi+\pi/2,\;\;\;n_{0k}\in\mbb{Z}, \quad\forall k=1,2,3
\end{equation}
Similarly, taking $k=0$ and $j>0$ we have
\begin{equation}
\label{Eq:phase-cond-2}
\varphi-\phi_0+\phi_j+\varphi_j=n_{j0}\pi+\pi/2,\;\;\;n_{j0}\in\mbb{Z}, \quad\forall j=1,2,3
\end{equation}
From this we infer
\begin{equation}
\label{Eq:phase-cond-3} 
\phi_0-\phi_l=(n_{0l}-n_{l0})\pi/2, \quad\forall l=1,2,3
\end{equation}
and 
\begin{align}\label{Eq:phase-cond-5}
\varphi+\varphi_l=(n_{0l}+n_{l0}+1)\pi/2 , \quad\forall l=1,2,3.
\end{align}

Now we turn to $j,k>0$.  We have
\begin{align}
e^{i\varphi}e^{i(\phi_{j}-\phi_k)}&\bra{\Phi_j}V\sigma_i V^\dagger\otimes\mbb{I}\ket{\Phi_k}\notag\\
&=e^{i\varphi}e^{i(\phi_{j}-\phi_k)}\tr[V\sigma_i V^\dagger\sigma_k\sigma_j],\quad\forall i=1,2,3.
\end{align}
Since $\sigma_k\sigma_j=i\epsilon_{kjl}\sigma_l$, the RHS of the previous equation becomes
\begin{equation}
i\epsilon_{kjl}e^{i\varphi}e^{i(\phi_j-\phi_k)}\tr[V\sigma_i V^\dagger \sigma_l].
\end{equation}
Again by Proposition \ref{Prop:Paul-real}, for this to be real, we need
\begin{equation}\label{Eq:phase-cond-4}
\varphi+\phi_j-\phi_k+\varphi_l=n_{jkl}\pi+\pi/2,\;\;\;n_{jkl}\in\mbb{Z}.
\end{equation}
for any distinct triple $(j,k,l)$ of nonzero indices.  %By interchanging $j$ and $k$ and adding this to the previous line, we obtain 
%\begin{equation}\label{Eq:phase-cond-4.5}
%\phi_j-\phi_k=(n_{jkl}-n_{kjl})\pi/2.
%\end{equation}
Substituting Eq. \eqref{Eq:phase-cond-5} into \eqref{Eq:phase-cond-4} we get
\begin{align}\label{Eq:phase-cond-6}
\phi_j-\phi_k=n_{jkl}\pi-(n_{0l}+n_{l0})\pi/2 ,
\end{align}
and adding Eq. \eqref{Eq:phase-cond-3} to this yields
\begin{align}
\phi_0-\phi_l+(\phi_j-\phi_k)\in \{n\pi\;|\;n\in\mbb{Z}\}
\end{align}
for any distinct triples $j,k,l>0$.  Finally, by applying the relations of Eqns. \eqref{Eq:magic-phase1}--\eqref{Eq:magic-phase4}, we have
\begin{align}
\alpha&=\frac{1}{4}(\phi_0-\phi_1+\phi_2-\phi_3)\in \{n\pi/4\;|\;n\in\mbb{Z}\}\notag\\
-\beta&=\frac{1}{4}(\phi_0-\phi_2+\phi_3-\phi_1)\in \{n\pi/4\;|\;n\in\mbb{Z}\}\notag\\
\gamma&=\frac{1}{4}(\phi_0-\phi_3+\phi_1-\phi_2)\in \{n\pi/4\;|\;n\in\mbb{Z}\}.
\end{align}
Hence $\alpha,\beta,\gamma$ are all integer multiples of $\pi/4$.  By Proposition \ref{Prop:magic-2}, it follows that $\Omega$ is a Clifford gate.
\end{proof}

\subsection{Proof of Theorem \ref{Thm:2qubits}}
\newtheorem*{Thm:2qubits}{Theorem \ref{Thm:2qubits}}
\begin{Thm:2qubits}
Any two-qubit unitary can be performed under LOBC with probability $(1- 2^{-N})^3$ using a consumption of $8N+1$ ebits.
\end{Thm:2qubits}

\begin{proof}

We freely interchange the symbols $\{1,2,3\}\leftrightarrow\{x,y,z\}$ to denote the standard Pauli operators.
%\begin{align}
%\sigma_1&=\begin{pmatrix}0&1\\1&0\end{pmatrix},&\sigma_2&=%\begin{pmatrix}0&i\\-i&0\end{pmatrix},&\sigma_3&=\begin{pmatrix}%1&0\\0&-1\end{pmatrix}.
%\end{align}
We will also write the identity as $\sigma_0=\mbb{I}$.  The two-qubit controlled-not (CNOT) gate will be denoted as
\begin{align}
\overrightarrow{U}_{x}&=\op{0}{0}\otimes\mbb{I}+\op{1}{1}\otimes\sigma_1.
\end{align}
In addition, we define the single-qubit matrices
\begin{align}
H&=\frac{1}{\sqrt{2}}\begin{pmatrix}1&1\\1&-1\end{pmatrix}\\
R_z(\theta)&=\begin{pmatrix}e^{i\theta/2}&0\\0&e^{-i\theta/2}\end{pmatrix},
\end{align}
as well as the two-qubit unitary
\begin{equation}
T_z(\theta)=R_z(-\theta)\oplus R_z(\theta)=e^{-i\theta\sigma_z\otimes\sigma_z/2}.
\end{equation}
Observe the relations
\begin{subequations}
\begin{align}
\label{Eq:commute1}
T_z(\theta)(\sigma_i\otimes\mbb{I})&=(\sigma_i\otimes\mbb{I})T_z(\theta)\quad\text{for $i=0,3$}\\
\label{Eq:commute2}
T_z(\theta)(\sigma_i\otimes\mbb{I})&=(\sigma_i\otimes\mbb{I})T_z(-\theta)\quad\text{for $i=1,2$}\\
\label{Eq:commute3}
T_z(\theta)(\mbb{I}\otimes\sigma_i)&=(\mbb{I}\otimes\sigma_i)T_z(\theta)\quad\text{for $i=0,3$}\\
\label{Eq:commute4}
T_z(\theta)(\mbb{I}\otimes\sigma_i)&=(\mbb{I}\otimes\sigma_i)T_z(-\theta)\quad\text{for $i=1,2$}.
\end{align}
\end{subequations}
From Propositions \ref{Prop:magic-1} and \ref{Prop:magic-3}, pre- and post- local unitaries can convert a given $U$ into an operator $\Omega$, which in the magic basis has the diagonal form
\begin{equation}
\diag[e^{i(\alpha-\beta+\gamma)},e^{i(-\alpha+\beta+\gamma)},e^{i(\alpha+\beta-\gamma)},e^{i(-\alpha-\beta-\gamma)}].
\end{equation}
The magic basis can then be rotated into the computational basis using a CNOT gate and local unitaries.  Doing so allows us to decompose any two-qubit unitary into the form
\begin{equation}\label{eq:mdecomp}
M(\alpha,\beta,\gamma)=\overrightarrow{U}_{x}(H\otimes\mbb{I})T_z(\beta)(R_z(\alpha)\otimes R_z(\gamma))(H\otimes\mbb{I})\overrightarrow{U}_{x},
\end{equation}
up to pre- and post- local unitaries \cite{Vatan-2004a}. Thus it suffices to implement $M(\alpha,\beta,\gamma)$ using LOBC. Similar to Protocol U2E, Protocol U2 relies heavily on the subroutine teleportation$^*$. Recall that teleportation$^*$ is standard teleportation using a maximally entangled two-qubit state without the classical communication and Pauli correction at the end. \\

\noindent\textbf{Protocol U2: LOBC implementation of $M(\alpha,\beta,\gamma)$:}

\begin{remark} Prior to Step 5 b, all operations by Alice (resp. Bob) will depend only on her (resp. his) previous measurement outcomes.
\end{remark}

$\bullet$ Input an arbitrary two-qubit state $\ket{\psi}^{AB}$.

\medskip

\noindent\textbf{Step 1 - Implement $(H\otimes\mbb{I})\overrightarrow{U}_{x}$:}
\begin{enumerate}
\item[{}] Using 1 ebit, Alice and Bob implement CNOT using the protocol given in Theorem \ref{Thm:cHermitian}, except they do not communicate their measurement outcomes to each other.  Alice then performs a Hadamard gate.  This leaves Alice (A) and Bob (B) sharing the state
\begin{align}
(\sigma_b\otimes\sigma_a)(H\otimes\mbb{I})\overrightarrow{U}_{x}\ket{\psi}^{AB}=:(\sigma_b\otimes\sigma_a)\ket{\psi_1}^{AB},
\end{align}
where $\sigma_a$ (resp. $\sigma_b$) is a Pauli error known to Alice (resp. Bob).  Note that $a,b\in\{0,1\}$.
\end{enumerate}

\medskip

\noindent\textbf{Step 2 - Implement $\mbb{I}\otimes R_z(\gamma)$:}
\begin{enumerate}
\item[a.] Initialize round $r=1$.  On system $B$, Bob performs $R_z(\gamma)$.  Using ebit $\ket{\Phi^+}^{A_1B_1}$, he then teleports$^*$ system $B$ to Alice, which leaves her in the state
\begin{equation}
(\sigma_b\otimes[\sigma_{b_1}R_z(\gamma)\sigma_a])\ket{\psi_1}^{AA_1}.
\end{equation}
\item[b.] On system $A_1$, Alice applies $\sigma_a$, and she enters the halting subroutine (see below) if $a\in\{0,3\}$.  Otherwise, using ebit $\ket{\Phi^+}^{A_2B_2}$ she teleports$^*$ system $A_1$ to Bob.  The resulting shared state is given by
\begin{equation}
(\sigma_b\otimes[\sigma_{a_2}\sigma_{b_1}R_z(-\gamma)])\ket{\psi_1}^{AB_2}.
\end{equation}
\item[c.] This begins round $r=2$.  If $b_1\in\{0,3\}$, Bob applies $R_z(2\gamma)$ to system $B_2$.  If $b_1\in\{1,2\}$ he applies $R_z(-2\gamma)$.  Using ebit $\ket{\Phi^+}^{A_3B_3}$, system $B_2$ is teleported$^*$ back to Alice.  This leaves them in the state
\begin{equation}
(\sigma_b\otimes\left[\sigma_{b_3}\sigma_{b_1}R_z(2\gamma)\sigma_{a_2}R_z(-\gamma)\right])\ket{\psi_1}^{AA_3}.
\end{equation}
\item[d.] On system $A_3$, Alice applies $\sigma_{a_2}$ and she enters the halting subroutine if $a_2\in\{0,3\}$.  Otherwise, using ebit $\ket{\Phi^+}^{A_4B_4}$, she teleports$^*$ system $A_3$ to Bob.  The resulting shared state is given by
\begin{equation}
(\sigma_b\otimes\left[\sigma_{a_4}\sigma_{b_3}\sigma_{b_1}R_z(-3\gamma)\right])\ket{\psi_1}^{AB_4}.
\end{equation} 
\item[e.] This continues for $N$ total rounds.  In each round, Bob applies either a positive or negative rotation with twice the magnitude of the rotation in the previous round.  Whether the rotation is positive or negative depends on the product of all his previous Pauli errors.

At the end of $N$ rounds, Alice will have entered the halting subroutine in some round $1\leq K\leq N$ with probability $1-2^{-N}$. If she entered in round $K$, then the state held on Alice's side at the start of the halting subroutine is
\begin{equation}
\label{Eq:Before_Halting}
(\sigma_b\otimes\left[\prod_{j=0}^{K-1}\sigma_{b_{2j+1}}R_z(\gamma)\right])\ket{\psi_1}^{AA_{2K-1}},
\end{equation}
 and the joint state at the end of $N$ rounds is
\begin{equation}
\label{Eq:Halting_Final}
(\sigma_b\otimes\left[\sigma_{a_{2N}}\prod_{j=K}^{N-1}\sigma_{b_{2j+1}'}\prod_{j=0}^{K-1}\sigma_{b_{2j+1}}R_z(\gamma)\right])\ket{\psi_1}^{AB_{2N}},
\end{equation}
where the $\sigma_{b_{2j+1}'}$ are the Pauli errors introduced by Alice for each round after she halted and $\sigma_{a_{2N}}$ is the teleportation$^*$ error from end of the halting subroutine.  If Alice never entered the halting subroutine, then at the end of $N$ rounds Alice and Bob's state is given by
\begin{equation}
\label{Eq:No_Halting_Final}
(\sigma_b\otimes\left[\sigma_{a_{2N}}\prod_{j=0}^{N-1}\sigma_{b_{2j+1}}R_z(-(2^N-1)\gamma)\right])\ket{\psi_1}^{AB_{2N}}.
\end{equation}

\item[f.] Bob applies to system $B_{2N}$ the concatenation of all his Pauli errors $\sigma_{\mbf{b}}:=\prod_{j=0}^{N-1}\sigma_{b_{2j+1}}$.  The crucial property of this protocol is that 
\begin{equation}
\prod_{j=0}^{N-1}\sigma_{b_{2j+1}}\left(\prod_{j=K}^{N-1}\sigma_{b_{2j+1}'}\prod_{j=0}^{K-1}\sigma_{b_{2j+1}}\right)\in\{\mbb{I},\sigma_z\}
\end{equation}
for any halting round $K$.  This holds because in the halting subroutine, Alice is able to distinguish whether Bob's teleportation error belongs to either $\{\mbb{I},\sigma_z\}$ or $\{\sigma_x,\sigma_y\}$.

\item[g.] If either Alice entered the halting subroutine during some round or $\gamma=l2^{-N}\pi$ (by Corollary \ref{corr:binaryAngles}), where $l$ is an even integer, then Alice and Bob's final shared state has the form
\begin{align}
\label{Eq:Round2_complete}
\nonumber(\sigma_b\otimes[\sigma^{\nu}_z\sigma_{a_{2N}}R_z(\gamma)])&\ket{\psi_1}^{AB_{2N}}\\
&=:(\sigma_b\otimes[\sigma^{\nu}_z\sigma_{a_{2N}}])\ket{\psi_2}^{AB_{2N}},
\end{align}
where $\nu\in\{0,1\}$ is a function of Bob's Pauli errors and Alice's halting round number.  The total ebit consumption in round 2 is $2N$.

\end{enumerate}

\begin{remark}Operations by Bob in Step 2 do not depend on whether Alice has entered the Halting Subroutine. 
\end{remark}
\begin{remark}  Here, $a_{2j}$ (resp. $b_{2j+1})$ represents Pauli errors induced in teleportation$^*$ by Alice using $\ket{\Phi^+}^{A_{2j}B_{2j}}$ (resp. by Bob using $\ket{\Phi^+}^{A_{2j+1}B_{2j+1}}$).
\end{remark}

\medskip

\noindent\textbf{Halting Subroutine:}

Suppose that Alice enters the halting subroutine in round $K$.  For each $K\leq j< N$:
\begin{enumerate}
\item[a.] Alice measures her half of ebit $\ket{\Phi^+}^{A_{2j}B_{2j}}$ in the computational basis.  This collapses system $B_{2j}$ into either $\ket{0}$ or $\ket{1}$.
\item[b.]  In round $j+1$, Bob applies either $R_z(2^{j}\gamma)$ or $R_z(-2^{j}\gamma)$ to system $B_{2j}$, as he would do had Alice not entered the halting subroutine.  Since $\ket{0}$ and $\ket{1}$ are both eigenstates of $R_z(2^{j}\gamma)$ and $R_z(-2^{j}\gamma)$, system $B_{2j}$ remains unchanged during this step.
\item[c.] Bob teleports$^*$ system $B_{2j}$ to Alice using ebit $\ket{\Phi^+}^{A_{2j+1}B_{2j+1}}$.  Alice's state in system $A_{2j+1}$ will be either $\sigma_{b_{2j+1}}\ket{0}$ or $\sigma_{b_{2j+1}}\ket{1}$.
\item[d.] Alice measures system $A_{2j+1}$ in the computational basis and can determine if $b_{2j+1}\in\{0,3\}$ or $b_{2j+1}\in\{1,2\}$ by comparing the measurement result to step a. %Alice measures system $A_{2j+1}$ and can determine if $b_{2j+1}\in\{0,3\}$ or $b_{2j+1}\in\{1,2\}$ based on whether a bit flip occurs.
\item[e.] If a bit flip occurs, Alice defines $b_{2j+1}'=1$ and she applies $\sigma_1$ to system $A_{2K-1}$.  If no bit flip occurs, she does nothing to this system and defines $b_{2j+1}'=0$.
\end{enumerate}
When these steps have been completed for all $K\leq j<N$, Alice uses $\ket{\Phi}^{A_{2N}B_{2N}}$ to teleport$^*$ system $A_{2K-1}$ to Bob.

\medskip

\noindent\textbf{Step 3: - Implement $R_z(\alpha)\otimes\mbb{I}$:}

Starting from the state in Eq. \eqref{Eq:Round2_complete}, Alice and Bob repeat Step 2 except with the roles reversed and with gate $R_z(\alpha)$ applied to the first system.  This leads to a state of the form
\begin{align}
\label{Eq:Round3_complete0}
\nonumber([\sigma^{\mu}_z\sigma_{b_{2N}}R_z(\alpha)]\otimes[\sigma^{\nu}_z\sigma_{a_{2N}}])\ket{\psi_2}^{A'_{2N}B_{2N}}\\
=:([\sigma^{\mu}_z\sigma_{b_{2N}}]\otimes[\sigma^{\nu}_z\sigma_{a_{2N}}])\ket{\psi_3}^{A'_{2N}B_{3N}},
\end{align}
with $\mu\in\{0,1\}$ being a function of Alice's Pauli errors and Bob's halting round.  For convenience, we will relabel systems $A'_{2N}$ and $B_{2N}$, as well as the Pauli errors, so that the state at the end of Step 3 is simply denoted by
\begin{align}
\label{Eq:Round3_complete}
([\sigma^{\mu}_z\sigma_{b}]\otimes[\sigma^{\nu}_z\sigma_{a}])\ket{\psi_3}^{AB}.
\end{align}
This step uses 2N ebits.

\medskip

\noindent\textbf{Step 4: - Implement $T_z(\beta)$:}

\begin{enumerate}
\item[a.] Initialize round $r=1$.  Starting from Eq. \eqref{Eq:Round3_complete}, Alice teleports$^*$ system $A$ to Bob using the shared ebit $\ket{\Phi^+}^{A_1B_1}$.
\item[b.]  Bob applies $\sigma_b$ to system $B_1$ and $T(\beta)$ across systems $BB_1$.  He teleports$^*$ both systems to Alice using ebits $\ket{\Phi^+}^{A_2B_2}\ket{\Phi^+}^{A_3B_3}$.  The resulting state in Alice's systems has the form
\begin{equation}
([\sigma^{\mu}_z\sigma_{b_2}]\otimes [\sigma^{\nu}_z\sigma_{b_3}])T_z(\beta)(\sigma_{a_1}\otimes\sigma_a)\ket{\psi_3}^{A_2A_3}.
\end{equation}
Note, crucially, that the $\sigma_z^\mu\otimes\sigma_z^\nu$ errors commute with $T_z(\beta)$.  
\item[c.] Alice applies $(\sigma_{a_1}\otimes\sigma_a)$ to systems $A_2A_3$.  If $(\sigma_{a_1}\otimes\sigma_a)$ commutes with $T_z(\beta)$ Alice halts and does nothing more for all future rounds; this occurs with probability $1/2$ and the halted state is given by
\begin{equation}
([\sigma^{\mu}_z\sigma_{b_2}]\otimes [\sigma^{\nu}_z\sigma_{b_3}])T_z(\beta)\ket{\psi_3}^{A_2A_3}.
\end{equation}
Otherwise the state is $([\sigma^{\mu}_z\sigma_{b_2}]\otimes [\sigma^{\nu}_z\sigma_{b_3}])T_z(-\beta)\ket{\psi_3}^{A_2A_3}$, and Alice proceeds to the next round.
\item[d.] This begins round $r=2$.  Given that Alice did not halt in the previous round,  she teleports$^*$ both systems $A_2A_3$ back to Bob using ebits $\ket{\Phi^+}^{A_4B_4}\ket{\Phi^+}^{A_5B_5}$.  His resulting state is
\begin{equation}
([\sigma^{\mu}_z\sigma_{a_4}\sigma_{b_2}]\otimes [\sigma^{\nu}_z\sigma_{a_5}\sigma_{b_3}])T_z(-\beta)\ket{\psi_3}^{B_4B_5}.
\end{equation}
\item[e.] Bob applies $T_z(2\beta)(\sigma_{b_2}\otimes\sigma_{b_3})$ to systems $B_4B_5$ and teleports$^*$ them back to Alice using ebits $\ket{\Phi^+}^{A_6B_6}\ket{\Phi^+}^{A_7B_7}$.
\item[f.] Alice applies $(\sigma_{a_4}\otimes\sigma_{a_5})$ to systems $A_6A_7$.  With probability $1/2$, $(\sigma_{a_4}\otimes\sigma_{a_5})$ commutes with $T_z(2\beta)$, and in which case Alice does nothing more for all future rounds.  Otherwise she proceeds to the next round.
\item[g.] This is continued for $N$ total rounds, each time Bob applying either a positive or negative $T_z(\theta)$ rotation with magnitude twice the magnitude of the rotation in the previous round.
\item[h.] At the end of $N$ rounds, Alice holds the state
\begin{align}
\label{Eq:step4-1}
\nonumber([\sigma^{\mu}_z\sigma_{b_{4K-2}}]\otimes [\sigma^{\nu}_z\sigma_{b_{4K-1}}])T_z(\beta)\ket{\psi_3}^{A_{4K-2}A_{4K-1}}\\
=:([\sigma^{\mu}_z\sigma_{b_{4K-2}}]\otimes [\sigma^{\nu}_z\sigma_{b_{4K-1}}])\ket{\psi_4}^{A_{4K-2}A_{4K-1}}
\end{align}
if she halted in round $1\leq K\leq N$, which occurs with probability $1-2^{-N}$.  Otherwise, she holds the state
\begin{align}
\label{Eq:step4-2}
\nonumber([\sigma^{\mu}_z\sigma_{b_{4N-2}}]\otimes &[\sigma^{\nu}_z\sigma_{b_{4N-1}}])\\
&*[T_z(-(2^N-1)\beta)\ket{\psi_3}^{A_{4N-2}A_{4N-1}}].
\end{align}
If, $\beta=l2^{-N}\pi$, where $l$ is an even integer, then Eq. \eqref{Eq:step4-2} is equivalent to Eq. \eqref{Eq:step4-1} with $K=N$.  In total, Step 4 uses $4N-1$ ebits.
\end{enumerate}

\begin{remark}  The labels of subsystems have been reset in moving from Step 3 to Step 4.  For example, the ebit $\ket{\Phi^+}^{A_1B_1}$ used in part a of Step 4 is distinct from the ebit used in part a of Step 2.  This is done to prevent an overload of notation.
\end{remark}

\medskip

\noindent\textbf{Step 5: - Implement $\overrightarrow{U}_{x}(H\otimes\mbb{I})$:}

\begin{enumerate}
\item[a.] Starting with Eq. \eqref{Eq:step4-1}, Alice holds the entire state. Since all local Pauli errors commute with $\overrightarrow{U}_{x}(H\otimes\mbb{I})$, Alice just applies this unitary directly. This generates a state that is equivalent to 
\begin{align}
\nonumber\overrightarrow{U}_{x}(H\otimes\mbb{I})\ket{\psi_4}&^{A_{4K-2}A_{4K-1}}\\
&=M(\alpha,\beta,\gamma)\ket{\psi}^{A_{4K-2}A_{4K-1}}
\end{align}
up to local Pauli errors. Alice teleports$^*$, system $A_{4K-1}$ back to Bob.
\item[b.] Alice and Bob communicate all previous measurement outcomes and halting rounds to one another.  Using this information, the local Pauli errors can be corrected on the previous state. Step 5 uses 1 ebit.
\end{enumerate}
%Given Eq. \eqref{Eq:step4-1} is attained, Alice teleports$^*$ system $A_{4K-1}$ to Bob using shared ebit $\ket{\Phi^+}^{A'B'}$.  Thus at the end of Step 4 they end up holding
%\begin{equation}
%\label{Eq:step4-final}
%([\sigma^{\mu}_z\sigma_{b_{4K-2}}]\otimes [\sigma^{\nu}_z\sigma_{a'}\sigma_{b_{4K-1}}])T_z(\beta)\ket{\psi_3}^{A_{4K-2}B'}=:([\sigma^{\mu}_z\sigma_{b_{4K-2}}]\otimes [\sigma^{\nu}_z\sigma_{a'}\sigma_{b_{4K-1}}])\ket{\psi_4}^{A_{4K-2}B'}.
%\end{equation}
%In total, Step 4 uses $4N$ ebits.
%\end{enumerate}

%\medskip

%\noindent\textbf{Step 5: - Implement $\overrightarrow{U}_{x}(H\otimes\mbb{I})$:}

%\begin{enumerate}
%\item[a.] Using 1 ebit, Alice and Bob implement the H+CNOT protocol (described below) on the state described in Eq. \eqref{Eq:step4-final}.  Since all local Pauli errors commute with $\overrightarrow{U}_{x}(H\otimes\mbb{I})$, this generates a state that is equivalent to 
%\begin{equation}
%\overrightarrow{U}_{x}(H\otimes\mbb{I})\ket{\psi_4}=M(\alpha,\beta,\gamma)\ket{\psi}
%\end{equation}
%up to local Pauli errors.
%\item[b.] Alice and Bob commute all previous measurement outcomes and halting rounds to one another.  Using this information, the local Pauli errors can be corrected on the previous state.
%\end{enumerate}

\end{proof}

Looking at step two of protocol U2, every failed rotation results in a rotation in the opposite direction.  We try to correct this by rotating with twice the angle of the previous step.  For certain unitaries $U(\alpha,\beta,\gamma)$ this leads to an implementation with probability one.
\begin{corollary}\label{corr:binaryAngles}
For any two-qubit unitary with $\alpha=l2^{-(N-1)}\pi$, $\beta=m2^{-(N-1)}\pi$, and $\gamma=p2^{-(N-1)}\pi$, where are $l,m,\text{ and }p$ are integers, $U(\alpha,\beta,\gamma)$ can be implemented deterministically (certainity) using LOBC with protocol U2.
\end{corollary}
\begin{proof}
Let us examine the proof for $\alpha$. The proof for the other two angles are the same. In step 2 of protocol U2, if Alice never enters the halting subroutine, then from \eqref{Eq:No_Halting_Final} we have
\begin{align}
R_z(-(2^N-1)\alpha)=R_z((-l2\pi+\alpha))=\pm R_z(\alpha).
\end{align}
\end{proof}

%\begin{remark}
%The requirement that $l,m,\text{ and }p$ be even comes from the way we defined $R_z(\theta)$. In $R_z(\theta)$, we divide $\theta$ by two in the exponentials. If we do not divide by two, then $l,m,\text{ and }p$ can be any integers.
%\end{remark}

\subsection{Proof of Theorem \ref{Thm:2bys}}
\newtheorem*{Thm:2bys}{Theorem \ref{Thm:2bys}}
\begin{Thm:2bys}
Let
\begin{equation}
U_\tau=\sum_{j=0}^{s-1} e^{i\tau_j}\op{j}{j}
\end{equation}
have phase angles $\tau_j\in[0,2\pi)$ such that $\tau_k\not=\tau_l$ for all $k\not=l\in\{0,\cdots,s-1\}$.  An LOBC implementation of the controlled unitary
\begin{equation}
U_c=\op{0}{0}\otimes \mbb{I}_s+\op{1}{1}\otimes U_\tau
\end{equation}
on a $2\otimes s$ system requires at least $\log s$ ebits of shared entanglement resource.
\end{Thm:2bys}
\begin{proof}
First, note that we can assume the entangled resource is a pure state $\ket{\eta}$.  The reason is that we are considering the exact implementation of $U_\tau$.  If a mixed-state resource was used, then the simulation of $U_\tau$ would need to be carried out for every single pure state in the mixture.  Hence one could dispense of the mixture and just exclusively use the pure state of lowest entanglement in the mixture.  The entanglement of this pure state places a lower bound on the entanglement of the original mixed state for any convex-roof extended entanglement measure, like the entanglement of formation \cite{Horodecki-2009a}.
 
A general LOBC protocol can be characterized by a local measurement for Alice and Bob, with Kraus operators $\{A_a\}_{a\in\mc{A}}$ and $\{B_b\}_{b\in\mc{B}}$ respectively, along with families of local unitaries, $\{W_{a,b}\}_{a\in\mc{A},b\in\mc{B}}$ for Alice and $\{V_{a,b}\}_{a\in\mc{A},b\in\mc{B}}$ for Bob.  The protocol will successfully simulate $U_c$ using a $d$-dimensional resource state $\ket{\eta}:=(\mbb{I}\otimes\hat{\eta})\ket{\Phi^+_d}$ if and only if for every $a\in\mc{A}$ and $b\in\mc{B}$ it holds that
\begin{align}\label{Eq:controlUisim}
\nonumber \left(\mbb{I}^{A_0B_0}\otimes M^{AA'BB'\to AB}_{ab}\right)\ket{\Phi^+}^{A_0A}\ket{\Phi_s^+}^{B_0B}\ket{\eta}^{A'B'}\\
=\gamma_{a,b}\mbb{I}^{A_0B_0}\otimes U_c^{AB}\ket{\Phi^+}^{A_0A}\ket{\Phi_s^+}^{B_0B}
\end{align}
where $M^{AA'BB'\to AB}_{ab}=W_{a,b}A_{a}^{AA'\to A}\otimes V_{a,b}B_b^{BB'\to B}$.  The amplitude $|\gamma_{a,b}|^2$ is the probability that Alice obtains measurement outcome $a\in\mc{A}$ and Bob obtains $b\in\mc{B}$.  To analyze further, it will be helpful to expand $A_a$ and $B_b$ in an orthonormal basis for system $A$ and $B$ respectively.  Doing so yields the general forms
\begin{align}\label{Eq:WAVBoperator}
\nonumber A_a =& \left(\sum\limits_{i=0}^1\ketbra{i}{0}^A\otimes\bra{\alpha_{0,i,a}}^{A'}+\sum\limits_{i=0}^1\ketbra{i}{1}^A\otimes\bra{\alpha_{1,i,a}}^{A'}\right)\\
\nonumber B_b =&\biggl(\sum\limits_{j=0}^{s-1}\ketbra{j}{0}^B\otimes\bra{\beta_{0,j,b}}^{B'}+\sum\limits_{j=0}^{s-1}\ketbra{j}{1}^B\otimes\bra{\beta_{1,j,b}}^{B'}\\
&+\cdots+\sum\limits_{j=0}^{s-1}\ketbra{j}{s-1}^B\otimes\bra{\beta_{s-1,j,b}}^{B'}\biggr)
\end{align}
where $\ket{\alpha_{i',i,a}}$ and $\ket{\beta_{j',j,b}}$ are both vectors in a $d$-dimensional space.  When expanded in the same basis, the RHS of Eq. \eqref{Eq:controlUisim} reads
\begin{align}\label{Eq:Uiaction}
\nonumber\gamma_{a,b}\mbb{I}^{A_0B_0}&\otimes U_c^{AB}\ket{\Phi^+}^{A_0A}\ket{\Phi_s^+}^{B_0B}=\dfrac{\gamma_{a,b}}{\sqrt{2s}}\biggl(\ket{00}^{A_0A}\\
&\otimes\sum_{j=0}^{s-1}\ket{jj}^{B_0B}+\ket{11}^{A_0A}\otimes\sum_{j=0}^{s-1}e^{i\tau_j}\ket{jj}^{B_0B}\biggr).
\end{align}
Thus, substituting \eqref{Eq:Uiaction} and \eqref{Eq:WAVBoperator} into \eqref{Eq:controlUisim} yields 
\begin{align}\label{Eq:substitute1}
\nonumber\sum_{i',i=0}^1\sum_{j',j=0}^{s-1}\ket{i'j'}^{A_0B_0}&\left(W_{a,b}^A\ket{i}^{A}\otimes V_{a,b}^B\ket{j}^{B}\bra{\beta_{j',j,b}}\hat{\eta}\ket{\alpha^*_{i',i,a}}\right)\\
\nonumber =&\gamma_{a,b}\sum_{j'=0}^{s-1}(\ket{0j'}^{A_0B_0}\otimes\ket{0j'}^{AB}\\
&+\ket{1j'}^{A_0B_0}\otimes e^{i\tau_{j'}}\ket{1j'}^{AB}),
\end{align}
where we use the relation $(\bra{\alpha_{i',i,a}}\otimes\bra{\beta_{j',j,b}})(\mbb{I}\otimes\hat{\eta})\ket{\Phi^+_d}=\bra{\beta_{j',j,b}}\hat{\eta}\ket{\alpha^*_{i',i,a}}$.  Eq. \eqref{Eq:substitute1} is equivalent to the system of equalities:
\begin{align}
\sum_{i,j}\ket{i}^A\ket{j}^B\bra{\beta_{0,j,b}}\hat{\eta}\ket{\alpha_{0,i,a}^*}&=\gamma_{a,b}W_{a,b}^\dagger\ket{0}^A\otimes V_{a,b}^\dagger\ket{0}^B\tag{E:$0$}\label{Eq:cons1s}\\
\sum_{i,j}\ket{i}^A\ket{j}^B\bra{\beta_{1,j,b}}\hat{\eta}\ket{\alpha_{0,i,a}^*}&=\gamma_{a,b}W_{a,b}^\dagger\ket{0}^A\otimes V_{a,b}^\dagger\ket{1}^B\tag{E:$1$}\label{Eq:cons2s}\\
&\vdots\notag\\
\nonumber\sum_{i,j}\ket{i}^A\ket{j}^B\bra{\beta_{s-1,j,b}}\hat{\eta}\ket{\alpha_{0,i,a}^*}&\\
=\gamma_{a,b}&W_{a,b}^\dagger\ket{0}^A\otimes V_{a,b}^\dagger\ket{s-1}^B\tag{E:$s-1$}\label{Eq:cons3s}\\
\nonumber\sum_{i,j}\ket{i}^A\ket{j}^B\bra{\beta_{0,j,b}}\hat{\eta}\ket{\alpha_{1,i,a}^*}&\\
=e^{i\tau_{0}}&\gamma_{a,b}W_{a,b}^\dagger\ket{1}^A\otimes V_{a,b}^\dagger\ket{0}^B\tag{F:$0$}\label{Eq:cons4s}\\
&\vdots\notag\\
\nonumber \sum_{i,j}\ket{i}^A\ket{j}^B\bra{\beta_{s-1,j,b}}\hat{\eta}\ket{\alpha_{1,i,a}^*}&\\
=e^{i\tau_{s-1}}&\gamma_{a,b}W_{a,b}^\dagger\ket{1}^A\otimes V_{a,b}^\dagger\ket{s-1}^B.\tag{F:$s-1$}\label{Eq:cons5s}
\end{align}
For any $k,k'\in\{0,\cdots,s-1\}$, take the outer products of Eqs. (E:$k$) and (E:$k'$), trace out system $A$, and sum over $a$.  Using the completion relation $\sum_{i,a}\op{\alpha^*_{0,i,a}}{\alpha^*_{0,i,a}}=\mbb{I}^{A'}$ we obtain
\begin{align}
\label{Eq:op1}
\sum_{j,j'}\op{j}{j'}^B\bra{\beta_{k,j,b}}\hat{\eta}\hat{\eta}^\dagger\ket{\beta_{k',j',b}}=\sum_{a}|\gamma_{a,b}|^2V^\dagger_{a,b}\op{k}{k'}V_{a,b}.
\end{align}
Performing the same calculation on Eqns. (F:$k$) and (F:$k'$) yields
\begin{align}
\label{Eq:op2}
\nonumber\sum_{j,j'}\op{j}{j'}^B\bra{\beta_{k,j,b}}\hat{\eta}\hat{\eta}^\dagger&\ket{\beta_{k',j',b}}\\
&=e^{i(\tau_k-\tau_{k'})}\sum_{a}|\gamma_{a,b}|^2V^\dagger_{a,b}\op{k}{k'}V_{a,b}.
\end{align}
From the assumption that $\tau_k\not=\tau_{k'}$ for $k\not=k'$, Eqs. \eqref{Eq:op1} and \eqref{Eq:op2} can both be true only if they are equaling zero; hence
\begin{equation}
\label{Eq:orthogonal_cond}
\bra{\beta_{k,j,b}}\hat{\eta}\hat{\eta}^\dagger\ket{\beta_{k',j',b}}=0\qquad\forall k\not=k',\;\forall j,j'\in\{0,\cdots,s-1\}.
\end{equation}
%At the same time, from Eqs. (E:$0$)--(E':$s-1$), for every $k\in\{0,\cdots,s-1\}$, there must be at least one value of $j$ such that $\ket{\beta_{k,j,b}}\not=0$.  Thus Eq. \eqref{Eq:orthogonal_cond} says there exists $s$ nonzero vectors that are pairwise orthogonal to each other.  This means that the $\{\ket{\beta_{k,j,b}}\}_{k,j}$ span an $s$-dimensional space, and consequently $s\leq d$.
We next define the operators
\begin{align}
\nonumber M_{b,t}=\frac{1}{s}\sum_{j=0}^{s-1}\sum_{k=0}^{s-1}\op{j}{\beta_{k,j,b}}e^{2\pi i tk/s},\\
b\in\mc{B},t\in\{0,1\cdots,s-1\}.
\end{align}
These, in fact, are Kraus operators for a complete measurement on system $B'$, as can be seen by
\begin{align}
\sum_{b,t}M_{b,t}^\dagger& M_{b,t}=\frac{1}{s^2}\sum_{b\in\mc{B}}\sum_{t,j,k,k'=0}^{s-1}\op{\beta_{k,j,b}}{\beta_{k',j,b}}e^{2\pi i t(k'-k)/s}\notag\\
&=\frac{1}{s^2}\sum_{b\in\mc{B}}\sum_{j,k,k'=0}^{s-1}\op{\beta_{k,j,b}}{\beta_{k',j,b}}\sum_{t=0}^{s-1}e^{2\pi i t(k'-k)/s}\notag\\
&=\frac{1}{s}\sum_{k=0}^{s-1}\sum_{b\in\mc{B}}\sum_{j=0}^{s-1}\op{\beta_{k,j,b}}{\beta_{k,j,b}}=\frac{1}{s}\sum_{k=0}^{s-1}\mbb{I}^{B'}=\mbb{I}^{B'}.
\end{align}
When this measurement is performed on $\hat{\eta}\hat{\eta}^\dagger$, we find
\begin{align}
\label{Eq:LOCC-reduce}
&M_{b,t}(\hat{\eta}\hat{\eta}^\dagger )M_{b,t}^\dagger\notag\\
&=\frac{1}{s^2}\sum_{j,j'=0}^{s-1}\sum_{k,k'=0}^{s-1}e^{2\pi i t(k-k')/s}\op{j}{j'}\bra{\beta_{k,j,b}}\hat{\eta}\hat{\eta}^\dagger\ket{\beta_{k',j',b}}\notag\\
&=\frac{1}{s^2}\sum_{k=0}^{s-1}\sum_{j,j'=0}^{s-1}\op{j}{j'}\bra{\beta_{k,j,b}}\hat{\eta}\hat{\eta}^\dagger\ket{\beta_{k,j',b}}\notag\\
&=\frac{1}{s}\sum_{a\in\mc{A}}|\gamma_{a,b}|^2\frac{\mbb{I}}{s},
\end{align}
where the second line follows from Eq. \eqref{Eq:orthogonal_cond} and the third line comes from setting $k=k'$ in Eq. \eqref{Eq:op1} and then summing over $k$ in both sides of that equation.  On the level of purifications, Eq. \eqref{Eq:LOCC-reduce} says that $(\mbb{I}^{A'}\otimes M^{B'}_{b,t})\ket{\eta}^{A'B'}$ is proportional to an $s$-dimensional maximally entangled state.  Since this holds for every outcome $M_{b,t}$, monotonicity of the entanglement entropy under local measurement implies that
\begin{equation}
\label{Eq:lowerbound}
\mathrm{E}(\ket{\eta})\geq \log s.
\end{equation} 
\end{proof}
\begin{remark}
The lower bound of Eq. \eqref{Eq:lowerbound} has been proven for the \textit{exact} implementation of $U_c$.  Thus, its significance lies in establishing the principle that LOBC requires more entanglement than LOCC for simulating certain gates, and this gap cannot be bounded even when fixing one of the systems to be a qubit.  To have true cryptographic application in tasks such as QPV, one would want a similar result for an $\epsilon$-approximate simulation of $U_c$.  We leave this to future work.
\end{remark}
\section{Conclusions}

\label{Sect:Conclusion} 

The LOBC setting is important in distributed quantum computing when time is of the essence. In this paper, we focused on the task of instantaneous nonlocal quantum computation, which is gate simulation using LOBC operations and pre-shared entanglement.  We have introduced a general two-qubit protocol that is exponentially better than other known protocols in terms of its entanglement consumption as a function of gate error.  We have shown this protocol to be non-optimal for the simulation of certain gates, such as swap, which can be implemented using just two ebits.  This two-ebit cost for swap is optimal even when interactive LOCC operations is permitted, two ebits are required for the implementation.  This is somewhat surprising given that swap is the most nonlocal two-qubit gate in the sense that it can generate the most entanglement, and it can be used for simultaneous message exchange between Alice and Bob.  Thus, our results suggest that the benefits of interactive communication in LOCC gate simulation mainly pertain to the entanglement cost of simulation rather than the entangling power of the simulated gate.

% We also introduced exact protocols for two-qubit unitaries that satisfy certain conditions. We then proved the surprising theorem that for a bipartite system with an arbitrary number of qubits for both parties (not necessarily equal) any controlled Hermitian unitary gate can be implemented with just one shared ebit. 

For a $2\otimes s$ system, we have shown that generic controlled unitary gates controlled from the $2$-dimensional side require at least $\log(s)$ ebits to implement. Currently we do not know whether this lower bound is close to achievable.  The known protocols have an ebit consumption that scales linearly with $s$ and some function of the error parameter, and it is an important open problem to determine if this exponential gap can be closed.  A more general theoretical question is whether every nonlocal gate can be perfectly implemented by LOBC using a \textit{finite} amount of entanglement.  Even in two-qubits, our new protocol has some failure probability unless $U(\alpha,\beta,\gamma)$ has special angles.  It is unknown if a protocol with no failure branches exists for every $U(\alpha,\beta,\gamma)$.

%This means that interactive communication allows an unbounded amount of savings in entanglement consumption. The same gap exists between non-Hermitian and Hermitian controlled unitary gates, where a Hermitian controlled unitary can always be performed with one ebit. An open area of research is to find why these gaps exist and more efficient protocols for arbitrary $D$ qubit unitary gates. Current general protocols require an exponential amount of entanglement as a function of the system size to implement.

\medskip

\section*{Acknowledgments}  We are extremely grateful to Barbara Kraus for providing helpful feedback and explaining various properties of multi-qubit unitaries.  We also thank the Centro de Ciencias de Benasque Pedro Pascual for hosting the 2018 multipartite entanglement workshop where an earlier version of this work was presented.  This work was supported by the Office of Naval Research Award No. N00014-15-12646.

\bibliographystyle{IEEEtran}
\bibliography{Broadcast_LOCC}

\end{document}